\documentclass[11pt]{article}

\usepackage{times}
\usepackage{amsmath}
\usepackage{algorithm}
\usepackage{amsthm}
\usepackage{graphicx,epsfig}
\usepackage{caption}
\usepackage[noend]{algpseudocode}
\theoremstyle{plain} \numberwithin{equation}{section}
\newtheorem{theorem}{Theorem}[section]
\newtheorem{corollary}[theorem]{Corollary}

\newtheorem{lemma}[theorem]{Lemma}

\theoremstyle{plain}

\newtheorem{property}[theorem]{Property}
\newtheorem{definition}[theorem]{Definition}

\newtheorem{assumption}[theorem]{Assumption}

\newtheorem{claim}{Claim}[section]
\advance \topmargin by -\headheight
\advance \topmargin by -\headsep
\evensidemargin \oddsidemargin
\usepackage{xcolor}
\usepackage{xspace}
\usepackage[shortlabels]{enumitem}

\usepackage[
style=alphabetic,
backref=true,
doi=false,
url=false,
maxcitenames=3,
mincitenames=3,
maxnames=10,
minbibnames=10,
backend=bibtex8,
sortlocale=en_US
]{biblatex}

\usepackage[ocgcolorlinks]{hyperref} 
\usepackage{cleveref}
\allowdisplaybreaks

\renewcommand{\paragraph}[1]{\medskip\noindent{\bf #1}\xspace}

\colorlet{DarkRed}{red!50!black}
\colorlet{DarkGreen}{green!50!black}
\colorlet{DarkBlue}{blue!50!black}

\hypersetup{
	linkcolor = DarkRed,
	citecolor = DarkGreen,
	urlcolor = DarkBlue,
	bookmarks = true,
	bookmarksnumbered = true,
	linktocpage = true
}

\begin{document}

\title{An Improved Algorithm for Incremental Cycle Detection and Topological Ordering in Sparse Graphs}

\author{Sayan Bhattacharya\thanks{University of Warwick, Coventry, UK. Email: {\tt S.Bhattacharya@warwick.ac.uk}} \and Janardhan Kulkarni\thanks{Microsoft Research, Redmond, USA. Email: {\tt jakul@microsoft.com}}}

\date{}

\maketitle

\begin{abstract}
We consider the problem of incremental cycle detection and topological ordering in a directed graph $G = (V, E)$ with $|V| = n$ nodes. In this setting, initially the edge-set $E$ of the graph is empty. Subsequently, at each time-step an edge gets inserted into $G$. After every edge-insertion, we have to report if the current graph contains a cycle, and as long as the graph remains acyclic, we have to maintain a topological ordering of the node-set $V$. Let $m$ be the total number of edges that get inserted into $G$. We present a randomized algorithm for this problem with $\tilde{O}(m^{4/3})$ total expected update time. 

Our result improves the $\tilde{O}(m \cdot \min (m^{1/2}, n^{2/3}))$ total update time bound of~\cite{BenderFGT16,HaeuplerKMST08,HaeuplerKMST12,CohenFKR13}. Furthermore, whenever $m = o(n^{3/2})$, our result improves upon the recently obtained $\tilde{O}(m \sqrt{n})$ total update time bound of~\cite{BernsteinC18}. We note that if $m = \Omega(n^{3/2})$, then the algorithm of~\cite{BenderFGT16,BenderFG09,CohenFKR13}, which has $\tilde{O}(n^2)$ total update time, beats the performance of the $\tilde{O}(m \sqrt{n})$ time algorithm of~\cite{BernsteinC18}. It follows that we improve upon the total update time of the algorithm of~\cite{BernsteinC18} in the  ``interesting'' range of sparsity where $m = o(n^{3/2})$.  

Our result also happens  to be the first one that breaks the $\Omega(n \sqrt{m})$ lower bound of~\cite{HaeuplerKMST08} on the total update time of any {\em local} algorithm for a nontrivial range of sparsity. Specifically, the total update time of our algorithm is $o(n \sqrt{m})$ whenever $m = o(n^{6/5})$. From a technical perspective, we obtain our result by combining the algorithm of~\cite{BernsteinC18} with the {\em balanced search} framework of~\cite{HaeuplerKMST12}. 
\end{abstract}

\thispagestyle{empty}

\newcounter{list}

\renewcommand{\L}{\mathcal{L}}

\newpage
\pagenumbering{arabic}

\section{Introduction}

Consider an {\em incremental}  directed graph $G = (V, E)$ with $|V| = n$ nodes. The edge-set $E$ is empty in the beginning. Subsequently, at each time step an edge gets inserted into $E$. After each such {\em update} (edge insertion), we have to report if the current graph $G$ contains a cycle, and as long as the graph remains acyclic, we have to maintain a topological ordering in $G$. The time taken to report the answer after an edge insertion is called the {\em update time}. We want to design an incremental algorithm for this problem with small {\em total update time}, which is defined as the sum of the update times over all the edge insertions. Recall that in the static setting there is an algorithm for cycle detection and topological ordering that runs in linear time. Thus, in the incremental setting, a naive approach would be to run this static algorithm from scratch after every edge-insertion in $G$. Let $m$ be the number of edges in the final graph. Then the naive incremental algorithm will have a total update time of $O(m \times (m+n)) = O(m^2+mn)$. In contrast, we get the following result.

\begin{theorem}
\label{th:main}
There is a randomized algorithm for incremental cycle detection with expected total update time of $\tilde{O}(m^{4/3})$.
\end{theorem}

\subsection{Perspective}

Cycle detection and topological ordering in directed graphs are fundamental, textbook problems. It is natural to ask what happens to the complexity of these problems when the input graph changes with time via a sequence of edge insertions. It comes as no surprise, therefore, that a long and influential line of work in the dynamic algorithms community, spanning over a couple of decades, have focussed on this question~\cite{BernsteinC18,HaeuplerKMST08,HaeuplerKMST12,BenderFGT16,BenderFG09,CohenFKR13,AjwaniF10,AjwaniFM08,KatrielB06,LiuC07,Marchetti-SpaccamelaNR96,PearceK06}. 

Note that the problem is trivial in the {\em offline setting}. Here,  we get an empty graph $G = (V, E)$ and a sequence of edges $e_1, \ldots, e_m$ as input {\em at one go}. For each $t \in [1, m]$, let $G_t$ denote the status of  $G$ after the first $t$ edges $e_1, \ldots, e_t$ have been inserted into $E$. We have to determine, for each $t$, if the graph $G_t$ contains a cycle. This offline version  can easily be solved in $O(m \log m)$ time using binary search. In contrast, we are still far away from designing an algorithm for the actual, incremental version of the  problem that has $\tilde{O}(m)$ total update time.\footnote{Throughout this paper, we use the $\tilde{O}(.)$ notation to hide polylog factors.}  This is especially relevant, because at present we do not know of any technique  in the conditional lower bounds literature~\cite{AbboudW14,HenzingerKNS15,KopelowitzPP16} that can prove a separation between the best possible total update time for an incremental problem and the best possible  running time for the corresponding offline version. Thus, although it might  be the case that there is no incremental algorithm for cycle detection and topological ordering with near-linear total update time, proving such a statement is beyond the scope of current techniques.  With this observation in mind, we now review  the current state of the art on the algorithmic front. We mention   three  results  that are particularly   relevant to this paper.

\smallskip
\noindent  {\bf Result (1):} There is an incremental algorithm with total update time of $\tilde{O}(n^2)$. This follows from the work of~\cite{BenderFGT16,BenderFG09,CohenFKR13}. So the problem is well understood for dense graphs where $m = \Theta(n^2)$.

\smallskip
\noindent  {\bf Result (2):} There is an incremental algorithm with total update time of $\tilde{O}(m \cdot \min(m^{1/2}, n^{2/3}))$. This follows from the work of~\cite{BenderFGT16,HaeuplerKMST08,HaeuplerKMST12,CohenFKR13}. 

\smallskip
\noindent {\bf Result (3):} There is a randomized incremental algorithm with total expected update time of $\tilde{O}(m \sqrt{n})$. This follows from the very recent work of~\cite{BernsteinC18}. 

\medskip
\noindent {\bf Significance of  Theorem~\ref{th:main}.} We obtain a randomized incremental algorithm for cycle detection and topological ordering that has an expected total update time of $\tilde{O}(m^{4/3})$. Prior to this, all incremental algorithms for this problem had a total update time of $\Omega(n^{3/2})$ for sparse graphs with $m = \Theta(n)$. Our algorithm breaks this barrier by achieving a bound of $\tilde{O}(n^{4/3})$ on sparse graphs. More generally, our total update time bound of $\tilde{O}(m^{4/3})$ outperforms  the $\tilde{O}(m \sqrt{n})$ bound from result (3) as long as $m = o(n^{3/2})$. Note that if $m = \omega(n^{3/2})$ then  result (3) gets superseded by  result (1). On the other hand,  result (3) is no worse than  result (2) for all values of $m$.\footnote{Throughout this paper we assume that $m \geq n$. This is because if $m = o(n)$ then many nodes remain isolated (with zero degree) in the final graph, and we can ignore those isolated nodes while analyzing the total update time of the concerned algorithm.}\footnote{It is easy to combine two incremental algorithms and get the ``best of both worlds". For example, suppose that we want to combine results (1) and (3) to get a total update time  of $\tilde{O}(\min(n^2, m \sqrt{n}))$,  without knowing the value of $m$ in advance. Then we can initially start with the algorithm from result (3) and then {\em switch to} the algorithm from result (1) when $m$ becomes $\Omega(n^{3/2})$.}  Thus, prior to our work result (3) gave the best known total update time when $m = o(n^{3/2})$, whereas  result (1) gave the best known total update time when $m = \Omega(n^{3/2})$. We now improve upon the bound from result (3) in this ``interesting" range of sparsity where $m = o(n^{3/2})$.

We are also able to break, for the first time in the literature, a barrier on the total update time of a certain type of  algorithms that was identified by  Haeupler et al.~\cite{HaeuplerKMST12}. Specifically, they defined an algorithm to be {\em local} iff it satisfies the following property. Suppose that currently the graph $G$ is acyclic, and the algorithm maintains a topological ordering $\prec$ on the node-set $V$ such that $x \prec y$ for every edge $(x, y) \in E$. In other words, every edge is a {\em forward edge} under $\prec$.  At this point,  a directed edge $(u, v)$ gets inserted into the graph $G$. Then the algorithm updates the topological ordering after this edge insertion only if $v \prec u$. Furthermore, if $v \prec u$, then the algorithm changes the positions of only those nodes in this topological ordering that lie in the {\em affected region}, meaning that a node $x$ changes its position only if $v \preceq x \preceq u$ just before the insertion of the edge.  Haeupler et al.~\cite{HaeuplerKMST12} showed that any local algorithm for incremental cycle detection and topological ordering must necessarily have a total update time of $\Omega(n \sqrt{m})$. Interestingly, although the algorithms that lead to results (1) and (3) are {\em not} local, prior to our work no algorithm (local or not) was known in the literature that beats this $\Omega(n \sqrt{m})$ lower bound for {\em any} nontrivial value of $m$. In sharp contrast, our algorithm (which is {\em not} local) has a total update time of $\tilde{O}(m^{4/3})$, and this beats the  $\Omega(n \sqrt{m})$ lower bound of Haeupler et al.~\cite{HaeuplerKMST12} when $m = o(n^{6/5})$.

\medskip
\noindent {\bf Our Technique.} We obtain our result by combining the framework of Bernstein and Chechik~\cite{BernsteinC18} with the balanced search procedure of Haeupler et al.~\cite{HaeuplerKMST12}. We first present a high level overview of the algorithm in~\cite{BernsteinC18}. Say that a node $x$ is an ancestor (resp. descendant) of another node $y$ iff there is a directed path from $x$ to $y$ (resp. from $y$ to $x$) in the current graph $G$. The algorithm in~\cite{BernsteinC18} is parameterized by an integer $\tau \in [1, n]$ whose value will be fixed later on. Initially, each node $v \in V$ is sampled with probability $\Theta(\log n/\tau)$.  Bernstein and Chechik~\cite{BernsteinC18} maintain a partition of the node-set $V$ into  subsets $\{ V_{i,j} \}$, where  a node $v$ belongs to a subset $V_{i,j}$ iff it has exactly $i$ ancestors and $j$ descendants among the sampled nodes. A total order $\prec^*$ is defined on the subsets $\{V_{i,j} \}$, where $V_{i,j} \prec^* V_{i', j'}$ iff either $i < i'$ or $\{i =i' \text{ and } j > j'\}$.  Next, it is shown that this partition and the total order satisfies  two important properties. (1) If $G$ contains a cycle, then w.h.p. all the nodes in that cycle belong to the same subset in the partition. (2) As long as $G$ remains acyclic, every edge $(u, v) \in E$ is either an {\em internal edge} or a {\em forward edge} w.r.t. the total order $\prec^*$; this means that the subset containing $u$ is either the same as or appears before the subset containing $v$.  Intuitively, these two properties allow us to {\em decompose} the problem into smaller parts. All we need to do now is (a)  maintain the subgraphs $G_{ij}$ induced by the subsets $V_{ij}$, and (b)  maintain a topological ordering within each subgraph $G_{i,j}$. Task (a) is  implemented by using an incremental algorithm for single-source reachability and a data structure for maintaining an {\em ordered list}~\cite{DietzS87}. 

For task (b), consider the scenario where an edge $(u, v)$ gets inserted and both $u$ and $v$ belong to the same subgraph $G_{i,j}$. Suppose that $u$ appears after $v$ in the current topological ordering in $G_{i,j}$. We now have to check if the insertion of the edge $(u, v)$ creates a cycle, or, equivalently, if there already exists a directed path from $v$ to $u$. In~\cite{BernsteinC18} this task is performed by doing a {\em forward search} from $v$. Intuitively, this means exploring the nodes  that are reachable from $v$ and appear  before $u$ in the current topological ordering. If we encounter the node $u$ during this forward search, then we have found the desired path from $v$ to $u$, and we can report that the insertion of the edge $(u, v)$ indeed creates a cycle. The time taken to implement this forward search is determined by the number of nodes $x$ that are explored during this search. Bernstein and Chechik~\cite{BernsteinC18} now introduce a crucial notion of {\em $\tau$-related pairs of nodes} (see Section~\ref{sec:prelim} for details), and show that for every node $x$ explored during the forward search we get a newly created $\tau$-related pair $(x, u)$. Next, they prove an upper bound of $O(n \tau)$ on the total number of such pairs that can appear throughout the duration of the algorithm. This implies that the total number of nodes explored during forward search is also at most $O(n \tau)$, and this in turn help us  fix the value of  $\tau$ (to balance the time taken for task (a)) and bound the total update time.

We now explain our main idea. Inspired by the balanced search technique from~\cite{HaeuplerKMST12}, we modify the subroutine for implementing task (b) as follows. We simultaneously perform a {\em forward search} from $v$ {\em and} a {\em backward search} from $u$. The forward search proceeds as in~\cite{BernsteinC18}. The backward search, on the other hand, explores the nodes $y$ such that $u$ is reachable from $y$ and $y$ appears before $v$ in the current topological ordering. We alternate between a forward search step and a backward search step, so that at any point in time the number of nodes respectively explored by these two searches are equal to one other. If these two searches {\em meet} at some node $z$, then we have found a path from $v$ to $u$ (the path goes via $z$), and we accordingly declare that the insertion of the edge $(u, v)$ creates a cycle. The time taken to implement task (b) is again determined by the number of nodes explored during the forward search, since this is the same as the number of nodes explored during the backward search. Now comes the following crucial observation. For every node $x$ explored during the forward search and every node $y$ explored during the backward search after the insertion of an edge $(u, v)$, we get a newly created $\tau$-related pair $(x, y)$. Thus, if $\lambda$ nodes are explored by each of these searches, then we get $\Omega(\lambda^2)$ newly created $\tau$-related pairs; although we still explore only $2\lambda$ nodes overall. In contrast, the algorithm in~\cite{BernsteinC18} creates only $O(\lambda)$ many new $\tau$-related pairs whenever it explores $\lambda$ nodes. This {\em quadratic improvement} in the creation of new $\tau$-related pairs leads to a much stronger bound on the total number of nodes explored by our algorithm, because as in~\cite{BernsteinC18} we still can have at most $O(n \tau)$ many newly created  $\tau$-related pairs during the entire course of the algorithm. This improved bound on the number of explored nodes leads to an improved bound of $\tilde{O}(m^{4/3})$ on the total update time.

\section{Our Algorithm: Proof of Theorem~\ref{th:main}}

This section is organized as follows.  In Section~\ref{sec:prelim} we recall some useful concepts from~\cite{BernsteinC18}. In Section~\ref{sec:algo} we  present our incremental algorithm, and in Section~\ref{sec:time} we analyze its total update time. The full version of the algorithm  (containing the proofs missing from the main body) appears in Appendix~\ref{app:sec:full}.

\subsection{Preliminaries}
\label{sec:prelim}

Throughout the paper, we assume that the maximum degree of a node in $G$ is at most $O(1)$ times the average degree. It was observed in~\cite{BernsteinC18} that this assumption is without any loss of generality. 
\begin{assumption}~\cite{BernsteinC18}
\label{assume}
Every node in $G$ has an out-degree of  $O(m/n)$ and an in-degree of  $O(m/n)$.
\end{assumption}

We say that a node $x \in V$ is an {\em ancestor} of another node $y \in V$ iff there is a directed path from $x$ to $y$ in $G$. We let $A(y) \subseteq V$ denote the set of all ancestors of $y \in V$. Similarly, we say that $x$ is a {\em descendant} of $y$ iff there is a directed path from $y$ to $x$ in $G$. We let $D(y) \subseteq V$ denote the set of all descendants of $y$. A node is both an ancestor and a descendant of itself, that is, we have $x \in A(x) \cap D(x)$. We also fix an integral parameter $\tau \in [1, n]$ whose exact value will be determined later on. Note that if there is a path from a node $x$ to another node $y$ in $G$, then  $A(x) \subseteq A(y)$ and $D(y) \subseteq D(x)$. Such a pair of nodes is said to be {\em $\tau$-related} iff the number of nodes in each of the sets $A(y) \setminus A(x)$ and $D(x) \setminus D(y)$ does not exceed $\tau$. 

\begin{definition}~\cite{BernsteinC18}
\label{def:related:nodes}
We say that an  ordered pair of nodes $(x, y)$ is {\em $\tau$-related} in the graph $G$ iff there is a path from $x$ to $y$ in $G$, and $|A(y) \setminus A(x)| \leq \tau$ and $|D(x) \setminus D(y)| \leq \tau$. We emphasize that for the ordered pair $(x, y)$ to be $\tau$-related, it is {\em not} necessary that there be an edge $(x,y) \in E$.
\end{definition}

If two nodes $x, y \in V$ are part of a cycle, then clearly $A(x) = A(y)$ and $D(x) = D(y)$, and both the ordered pairs $(x,y)$ and $(y, x)$ are $\tau$-related. In other words, if an ordered pair $(x, y)$ is not $\tau$-related, then there is no cycle containing both $x$ and $y$. Intuitively, therefore, the notion of $\tau$-relatedness serves as a {\em relaxation} of the notion of two nodes being part of a cycle. Next, note that the graph $G$ keeps changing as more and more edges are inserted into it. So it might  be the case that an ordered pair of nodes $(x, y)$ is {\em not} $\tau$-related in $G$ at some point in time, but {\em is} $\tau$-related in $G$ at some other point in time.  The following definition and the subsequent theorem becomes relevant in light of this observation.

\begin{definition}~\cite{BernsteinC18}
\label{def:related:nodes:time}
We say that an {\em ordered} pair of nodes $(x, y)$ is {\em sometime $\tau$-related} in the graph $G$ iff it is $\tau$-related at some point in time during the entire sequence of edge insertions in $G$.
\end{definition}

\begin{theorem}~\cite{BernsteinC18}
\label{th:bound:related:nodes}
The number of sometime $\tau$-related pairs of nodes in $G$ is at most $O(n \tau)$.
\end{theorem}

Following~\cite{BernsteinC18}, we  maintain a partition of the node-set $V$ into subsets $\{V_{i,j}\}$ and the subgraphs $\{G_{i,j} = (V_{i,j}, E_{i,j})\}$ induced by these subsets of nodes.  We sample each node  $x \in V$ independently with probability $\log n/\tau$. Let $S \subseteq V$ denote the set of these sampled nodes.  The outcome of this random sampling gives rise to a partition of the node-set $V$ into $(|S|+1)^2$ many subsets $\{V_{i,j}\}$, where $i,j \in [0, |S|]$. This is  formally defined as follows.  For every node $x \in V$, let $A_S(x) = A(x) \cap S$ and $D_S(x) = D(x) \cap S$ respectively denote the set of ancestors and descendants of $x$ that have been sampled. Each subset $V_{i,j} \subseteq V$  is indexed by an ordered pair $(i, j)$ where $i \in [0, |S|]$ and $j \in [0, |S|]$.  A node $x \in V$ belongs to a subset $V_{i, j}$ iff $|A_S(x)| = i$ and $|D_S(x)| = j$. In words, the index $(i, j)$ of the subset $V_{i,j}$ specifies the number of sampled ancestors and sampled descendants each node $x \in V_{i,j}$ is allowed to have. It is easy to check that the subsets $\{V_{i,j}\}$ form a valid partition the node-set $V$. Let $E_{i, j} = \{ (x, y) \in E : x, y \in V_{i,j} \}$ denote the set of edges in $G$ whose both endpoints lie in $V_{i,j}$, and let $G_{i, j} = (V_{i,j}, E_{i,j})$ denote the subgraph of $G$ induced by the subset of nodes $V_{i,j}$.  We also define a total order $\prec^*$ on the subsets  $\{V_{i,j}\}$, where we have $V_{i,j} \prec^* V_{i',j'}$ iff either  $\{i < i'\}$ or $\{i = i' \text{ and } j > j'\}$.  We slightly abuse the notation by letting  $V(x)$ denote the unique subset $V_{i,j}$ which contains the  node $x \in V$. Consider any edge $(x, y) \in E$. If the two endpoints of the edge belong to two different subsets in the partition $\{V_{i,j}\}$, i.e., if $V(x) \neq V(y)$, then we refer to the edge $(x, y)$ as a {\em cross edge}. Otherwise, if $V(x) = V(y)$, then  the edge $(x,y)$ is an {\em internal edge}.

\begin{lemma}~\cite{BernsteinC18}
\label{lm:propofpartition}
Consider the partition of the node-set $V$ into subsets $\{V_{i,j}\}$, and the subgraphs $\{G_{i,j} = (V_{i,j}, E_{i,j})\}$ induced by these subsets of nodes. They satisfy the following three properties.
\begin{itemize}
\item If there is a cycle in  $G = (V, E)$, then every edge of that cycle is an internal edge. 
\item For every cross edge $(x, y) \in E$, we have $V(x) \prec^* V(y)$.  
\item Consider any two nodes $x, y \in V_{i,j}$ for some $i, j \in [0, |S|]$. If there is a path from $x$ to $y$ in the subgraph $G_{i,j}$, then with high probability the ordered pair $(x, y)$ is $\tau$-related in $G$. 
\end{itemize}
\end{lemma}

\noindent The first property states that the graph $G$ contains a cycle iff some subgraph $G_{i,j}$ contains a cycle. Hence, in order to detect a cycle in $G$ it suffices to only consider the edges that belong to the induced subgraphs $\{G_{i,j}\}$. The second property, on the the other hand, implies that if the graph $G$ is acyclic, then it admits a topological ordering $\prec$ that is {\em consistent} with the total order $\prec^*$, meaning   that $x \prec y$ for all $x, y \in V$ with $V(x) \prec^* V(y)$. Finally, the last property states that whenever a subgraph $G_{i,j}$ contains  a path from a node $x$ to some other node $y$, with high probability  the ordered pair $(x,y)$ is $\tau$-related in the input graph $G$.

\subsection{The algorithm}
\label{sec:algo}

Since edges never get deleted from the graph $G$, our algorithm does not have to do anything once it detects a cycle (for the graph will continue to have a cycle after every edge-insertion in the future). Accordingly, we assume that the graph $G$ has  remained acyclic throughout the sequence of edge insertions till the present moment, and our goal is to check if  the next edge-insertion creates a cycle in $G$. Our algorithm maintains a topological ordering $\prec$ of the node-set $V$ in the graph $G$ that is {\em consistent} with the total order $\prec^*$ on the subsets of nodes $\{V_{i,j}\}$, as defined in Section~\ref{sec:algo}. Specifically, we maintain a {\em priority} $k(x)$ for every node $x \in V$, and for every two nodes $x, y \in V$ with $V(x) \prec^* V(y)$ we  ensure that $k(x) \prec k(y)$. As long as $G$ remains acyclic, the existence of such a topological ordering $\prec$ is guaranteed by Lemma~\ref{lm:propofpartition}.

\medskip
\noindent {\bf Data Structures.} We maintain the partition $\{V_{i,j}\}$ of the node-set $V$ and the subgraphs $\{G_{i,j} = (V_{i,j}, E_{i,j})\}$ induced by the subsets in this partition. We  use an {\em ordered list} data structure~\cite{DietzS87} on the node-set $V$ to implicitly maintain the priorities $\{k(x)\}$ associated with the topological ordering $\prec$. This data structure supports each of the following operations in $O(1)$ time. 
\begin{itemize}
\item INSERT-BEFORE($x, y$): This inserts the node $y$ just before the node $x$ in the topological ordering.
\item INSERT-AFTER($x, y$): This inserts the node $y$ just after the node $x$ in the topological ordering. 
\item DELETE($x$): This deletes the node $x$ from the existing topological ordering.
\item COMPARE($x, y$): If $k(x) \prec k(y)$, then this returns YES, otherwise this returns NO.
\end{itemize}

\noindent The implementation of our algorithm requires the creation of   two {\em dummy nodes} $x_{i,j}$ and $y_{i,j}$ in every subset $V_{i, j}$. We ensure that $k(x_{i,j}) \prec k(x) \prec k(y_{i,j})$ for all $x \in  V_{i, j}$. In words, the dummy node $x_{i,j}$ (resp. $y_{i,j}$) comes {\em first} (resp. {\em last}) in the topological order among all the nodes in $V_{i,j}$. Further, for all nodes $x \in V$ with $V(x) \prec V_{i,j}$ we have $k(x) \prec k(x_{i,j})$, and for all nodes $x \in V$ with $V_{i,j} \prec V(x)$ we have $k(y_{i,j}) \prec k(x)$.

\medskip
\noindent {\bf Handling the insertion of an edge $(u,v)$ in $G$.} By induction hypothesis, suppose that the graph $G$ currently does not contain any cycle and we are maintaining the topological ordering $\prec$ in $G$. At this point, an edge $(u,v)$ gets inserted into $G$. Our task now is to first figure out if the insertion of this edge creates a cycle, and if not, then to update the topological ordering $\prec$. We perform this task in four {\em phases}, as described below.
\begin{enumerate}
\item In phase I, we update the subgraphs $\{ G_{i,j}\}$. 
\item In phase II, we update the total order $\prec$  to make it consistent with the total order $\prec^*$. 
\item In phase III, we check if the edge-insertion  creates a cycle in $G$. See Section~\ref{sec:phase:cycle} for details.
\item If phase III fails to detect a cycle, then in phase IV we further update (if necessary) the total order $\prec$ so as to ensure that it is a topological order in the current graph $G$. See Section~\ref{sec:phase:final} for details.
\end{enumerate}
\noindent {\bf Remark.} We follow the framework developed in~\cite{BernsteinC18} while implementing Phase I and Phase II. We differ from~\cite{BernsteinC18} in Phase III and Phase IV, where we use the {\em balanced search} approach from~\cite{HaeuplerKMST12}.  

\medskip
\noindent {\bf Implementing Phase I.} In the first phase, we update the subgraphs $\{ G_{i,j}\}$ such that they satisfy the properties mentioned in Lemma~\ref{lm:propofpartition}. The  next lemma follows from~\cite{BernsteinC18}. The key idea is to maintain incremental single-source reachability data structures from each of the sampled nodes. Since at most $\tilde{O}(n/\tau)$ many nodes are sampled in expectation, and since each incremental single-source reachability data structure requires $\tilde{O}(m)$ total update time to handle $m$ edge insertions, we get the desired bound of $\tilde{O}(mn/\tau)$.

\begin{lemma}~\cite{BernsteinC18}
\label{lm:partition:time}
In phase I, the algorithm spends $\tilde{O}(m n/\tau)$ total update time in expectation.
\end{lemma}

\noindent {\bf Implementing Phase II.}
In this phase we update the total order $\prec$ on the node-set $V$ in a certain manner. Let $G^-$ and $G^+$ respectively denote the graph $G$ just before and just after the insertion of the edge $(u,v)$. Similarly, for every node $x \in V$,  let $V^-(x)$ and $V^+(x)$ respectively denote the subset $V(x)$ just before and just after the insertion of the edge $(u, v)$. At the end of this phase, the following properties are satisfied. 
\begin{property}~\cite{BernsteinC18}
\label{cor:phase:order}
At the end of phase II the total order $\prec$ on $V$ is consistent with the total order $\prec^*$  on  $\{V_{i,j}\}$. Specifically, for any two nodes $x$ and $y$, if $V(x) \prec^* V(y)$, then we also have $k(x) \prec k(y)$. 
\end{property}

\begin{property}~\cite{BernsteinC18}
\label{lm:phase:order:correct}
At the end of phase II the total order $\prec$ on $V$  remains a valid topological ordering of $G^-$, where $G^-$ denotes the graph $G$ just before the insertion of the edge $(u,v)$.
\end{property}

The next lemma bounds the total time spent by the algorithm in phase II. 

\begin{lemma}~\cite{BernsteinC18}
\label{lm:phase:order:time}
The total time spent in phase II across all edge-insertions is at most $\tilde{O}(n^2/\tau)$.
\end{lemma}

\begin{proof}(Sketch)
Let $C$ be a counter that keeps track of the number of times some node moves from one subset in the partition $\{V_{i, j}\}$ to another. Recall that  a node $x \in V$ belongs to a subset $V_{i,j}$ iff $|A_S(x)| = i$ and $|D_S(x)| = j$.  As more and more edges keep getting inserted in $G$, the node $x$ can never lose a sampled node in $S$ as its ancestor or descendent. Instead, both the sets $A_S(x)$ and $D_S(x)$ can only grow with the passage of time. Since $|A_S(v)|, |D_S(v)| \in [0, |S|]$, each node $x$ can move from one subset in the partition $\{V_{i,j}\}$ to another at most $2 \cdot |S|$ times. Thus,  we have $C \leq |V| \cdot 2 |S| = O(n |S|)$. Since $\mathbf{E}[|S|] = \tilde{O}(n/\tau)$, we conclude that $\mathbf{E}[C] = \tilde{O}(n^2/\tau)$. Now, phase II can be implemented in such a way that  a call is made to the ordered list data structure~\cite{DietzS87} only when some node moves from one subset of the partition $\{V_{i,j}\}$ to another. So the total time spent in phase II is at most  $C$, which happens to be $\tilde{O}(n^2/\tau)$ in expectation.
\end{proof}

\subsubsection{Phase III: Checking if the insertion of the edge $(u,v)$ creates a cycle.} 
\label{sec:phase:cycle}

Let $G^-$ and $G^+$ respectively denote the graph $G$  before and after the insertion of the edge $(u, v)$. Consider the total order $\prec$  on the set of nodes $V$ in the beginning of phase III (or, equivalently, at the end of phase II). Property~\ref{cor:phase:order} guarantees that  $\prec$ is consistent with the total order $\prec^*$ on $\{V_{i,j}\}$, and Property~\ref{lm:phase:order:correct} guarantees that $\prec$ is a valid topological ordering in $G^-$.  We will use these two properties  throughout the current phase. The pseudocodes of all the subroutines used in this phase appear in Section~\ref{sec:codes}.

In phase III, our goal is to determine if the insertion of the edge $(u, v)$ creates a cycle in $G$. Note that if  $k(u) \prec k(v)$, then   $\prec$ is also a valid topological ordering in $G^+$ as per Property~\ref{lm:phase:order:correct}, and clearly the insertion of the edge $(u,v)$ does not create a cycle.  The difficult case occurs when $k(v) \prec k(u)$. In this case, we first infer that $V(u) = V(u)$, meaning that both  $u$ and $v$ belong to the same subset in the partition $\{V_{i,j}\}$ at the end of phase II. This is because of the following reason. The total order $\prec$ is consistent with the total order $\prec^*$ as per Property~\ref{cor:phase:order}. Accordingly, since $k(v) \prec k(u)$, we conclude that if $V(v) \neq V(u)$ then  $V(v) \prec^* V(u)$. But this would contradict Lemma~\ref{lm:propofpartition} as there is a cross edge from $u$ to $v$.

To summarize, for the rest of this section we assume that $k(v) \prec k(u)$ and $V(v) = V(u) = V_{i,j}$ for some $i, j \in [0, |S|]$. We have to check if there is a path $P_{v,u}$ from $v$ to $u$ in $G^-$. Along with the edge $(u,v)$, such a path $P_{v,u}$ will define a cycle in $G^+$. Hence, by Lemma~\ref{lm:propofpartition}, every edge $e$ in such a path $P_{v, u}$ will belong to the subgraph $G_{i,j} = (V_{i,j}, E_{i,j})$. Thus, from now on our task is to determine if there is a path $P_{v,u}$ from $v$ to $u$ in $G_{i,j}$. We perform this task by calling the subroutine SEARCH($u,v$) described below.

\medskip
\noindent {\bf SEARCH($u,v$).} 
We  conduct two searches in order to find the path $P_{v, u}$: A {\em forward search} from $v$, and a {\em backward search} from $u$. Specifically, let $F$ and $B$ respectively denote the set of nodes visited by the forward search and the backward search till now. We always ensure that $F \cap B = \emptyset$. A node in $F$ (resp. $B$) is referred to as a forward (resp. backward) node.   Every forward node $x \in F$ is reachable from the node $v$ in $G_{i,j}$, whereas the node $u$ is reachable from every backward node $x \in B$ in $G_{i,j}$. We further classify each of the sets $F$ and $B$ into two subsets: $F_a \subseteq F$, $F_d = F \setminus F_a$ and  $B_a \subseteq B$, $B_d = B \setminus B_a$. The nodes in $F_a$ and $B_a$ are called {\em alive}, whereas the nodes in $F_d$ and $B_d$ are called {\em dead}. Intuitively, the dead nodes have already been {\em explored} by the search, whereas the alive nodes have not yet been explored.  When the subroutine begins execution, we have  $F_a = \{v\}$ and $B_a = \{u \}$. The following property is always satisfied.

\begin{property}
\label{prop:forward:back}
Every node $x \in F_a \cup F_d$ is reachable from the node $v$ in $G_{i,j}$, and the node $u$ is reachable from every node $y \in B_a \cup B_d$ in $G_{i,j}$. The sets $F_a, F_d, B_a$ and $B_d$ are pairwise mutually exclusive.
\end{property}

A simple strategy for exploring a  forward and alive node $x \in F_a$ is as follows. For each of its outgoing edges $(x, y) \in E_{i,j}$, we check if $y \in B$. If yes, then we have detected a path from $v$ to $u$: This path goes from $v$ to $x$ (this is possible since $x$ is a forward node), follows the edge $(x, y)$, and then from $y$ it goes to $u$ (this is possible since $y$ is a backward node).  Accordingly, we stop and report that the graph $G^+$ contains a cycle. In contrast, if $y \notin B$ and $y \notin F$, then we insert $y$ into the set $F_a$ (and $F$), so that $y$ becomes a forward and alive node  which will be explored in future. In the end, we move the node $x$ from the set $F_a$ to the set $F_d$. We refer to the subroutine that explores a node $x \in F_a$ as {\bf EXPLORE-FORWARD($x$)}.

Analogously, we explore a  backward and alive node $x \in B_a$ is as follows. For each of its incoming edges $(y, x) \in E_{i,j}$, we check if $y \in F$. If yes, then there is a path from $v$ to $u$: This path goes from $v$ to $y$ (this is possible since $y$ is a forward node), follows the edge $(y, x)$, and then from $x$ it goes to $u$ (this is possible since $x$ is a backward node).  Accordingly, we stop and report that the graph $G^+$ contains a cycle. In contrast, if $y \notin F$ and $y \notin B$, then we insert $y$ into the set $B_a$ (and $B$), so that $y$ becomes a backward and alive node  which will be explored in future. In the end, we move the node $x$ from the set $B_a$ to the set $B_d$. We refer to the subroutine that explores a node $x \in B_a$ as {\bf EXPLORE-BACKWARD($x$)}. 


\begin{property}
\label{prop:node:move}
Once a node $x \in F_a$ (resp. $x \in B_a$) has been explored, we delete it from the set $F_a$ (resp. $B_a$) and insert it into the set $F_d$ (resp. $B_d$).
\end{property}

While exploring a node $x \in F_a$ (resp. $x \in B_a$), we ensure that all its outgoing (resp. incoming) neighbors are included in $F$ (resp. $B$). This leads to the following important corollary.
\begin{corollary}
\label{cor:alive:dead}
Consider any edge $(x, y) \in E_{i,j}$. At any point in time, if $x \in F_d$, then at that time we  also have $y \in F_a \cup F_d$.  Similarly, at any point in time, if $y \in B_d$, then at that time we also have $x \in B_a \cup B_d$.
\end{corollary}

Two natural questions arise at this point. First, how frequently do we explore forward nodes compared to exploring backward nodes? Second, suppose that we are going to explore a forward (resp. backward) node at the present moment. Then how do we select the node $x$ from the set $F_a$ (resp. $B_a$) that has to be explored? Below, we  state two crucial properties of our algorithm that address these two questions. 

\begin{property}(Balanced Search)
\label{prop:balance}
We  alternate between calls to  EXPLORE-FORWARD(.) and  EXPLORE-BACKWARD(.).  This ensures that $|B_d| - 1 \leq |F_d| \leq  |B_d|+1$  at every point in time. In other words,  every forward-exploration step is followed by a backward-exploration step and vice versa.
\end{property}

\begin{property}(Ordered Search)
\label{prop:order}
While deciding which node in $F_a$ to explore next, we always pick the node $x \in F_a$ that has {\em minimum} priority $k(x)$. Thus,  we ensure that the subroutine EXPLORE-FORWARD($x$) is only called on the node $x$ that appears before every other node in $F_a$ in the total ordering $\prec$. In contrast, while deciding which node in $B_a$ to explore next, we always pick the node $y \in B_a$ that has {\em maximum} priority $k(y)$. Thus,  we ensure that the subroutine EXPLORE-BACKWARD($y$) is only called on the node $x$ that appears {\em after} every other node in $B_a$ in the total ordering $\prec$.
\end{property}

An immediate consequence of Property~\ref{prop:order} is that there is no {\em gap} in the set $F_d$ as far as reachability from the node $v$ is concerned. To be more specific, consider the sequence of nodes in $G_{i,j}$ that are reachable from $v$ in increasing order of their positions in the total order $\prec$. This sequence starts with $v$. The set of nodes belonging to $F_d$ always form a prefix of this sequence. This observation is formally stated below.

\begin{corollary}
\label{cor:order:forward}
Consider any two nodes $x, y \in V_{i,j}$ such that $k(x) \prec k(y)$ and there is a path in $G_{i,j}$ from $v$ to each of these two nodes. At any point in time, if $y \in F_d$, then we must also have $x \in F_d$. 
\end{corollary}

Corollary~\ref{cor:order:backward} is a mirror image of Corollary~\ref{cor:order:forward}, albeit from the perspective of the node $u$.

\begin{corollary}
\label{cor:order:backward}
Consider any two nodes $x, y \in V_{i,j}$ such that $k(x) \prec k(y)$ and there is a path in $G_{i,j}$ from each of these two nodes to $u$. At any point in time, if $x \in B_d$, then we must also have $y \in B_d$. 
\end{corollary}

To complete the description of the subroutine SEARCH($u,v$), we now specify six {\em terminating conditions}. Whenever one of these conditions is satisfied, the subroutine does not need to run any further because it already  knows whether or not the insertion of the edge $(u, v)$ creates a cycle in the graph $G$.

\smallskip
\noindent
{\bf (C1)} {\em  $F_a = \emptyset$}. 

\smallskip
\noindent
In this case, we conclude that the graph $G$ remains acyclic even after the insertion of the edge $(u,v)$. We now justify this conclusion.  Recall that if the insertion of the edge $(u,v)$ creates a cycle, then that cycle must contain a path $P_{v, u}$ from $v$ to $u$ in $G_{i,j}$. When the subroutine SEARCH($u,v$) begins execution, we have $F_a = \{v\}$ and $B_a = \{u\}$. Hence, Property~\ref{prop:node:move} implies that  at the present moment $v \in F_d \cup F_a$ and $u \in B_d \cup B_a$. Since the sets $F_d, F_a, B_d, B_a$ are pairwise mutually exclusive (see Property~\ref{prop:forward:back}) and  $F_a = \emptyset$, we currently have $v \in F_d$ and $u \notin F_d$. Armed with this observation, we consider the path $P_{vu}$ from $v$ to $u$, and let $x$ be the first node in this path that does not belong to  $F_d$. Let $y$ denote the node that appears just before $x$ in this path. Then by definition, we have $y \in F_d$ and $(y, x) \in E_{i,j}$. Now, applying Corollary~\ref{cor:alive:dead}, we get $x \in F_d \cup F_a = F_d$, which leads to a contradiction.

\smallskip
\noindent
{\bf (C2)} {\em  $B_a = \emptyset$}.

\smallskip
\noindent
This is analogous to the condition (C1) above, and we conclude that  $G$ remains acyclic in this case. 

\smallskip
\noindent 
 {\bf (C3)} {\em While exploring a node $x \in F_a$, we discover that $x$ has an outgoing edge to a node $x' \in B_a \cup B_d$.}

\smallskip
\noindent
Here,  we conclude that the insertion of the edge $(u,v)$ creates a cycle. 
We now justify this conclusion. Since $x \in F_a$, Property~\ref{prop:forward:back}  implies that there is a path $P_{v, x}$ from $v$ to $x$. Since $x' \in B_a \cup B_d$, Property~\ref{prop:forward:back} also implies that there is a path $P_{x', u}$ from $x'$ to $u$. We get a cycle by combining the path $P_{v, x}$, the edge $(x, x')$,  the path $P_{x', u}$ and the edge $(u,v)$.

\smallskip
\noindent
{\bf (C4)} {\em While exploring a node $y \in B_a$, we discover that $y$ has an incoming edge from a node $y' \in F_a \cup F_d$.}

\smallskip
\noindent
Similar to condition (C3), in this case we conclude that the insertion of the edge $(u,v)$ creates a cycle. 

\smallskip
\noindent
 {\bf (C5)} {\em  $\min_{x \in F_a} k(x) \succ \min_{y \in B_d} k(y)$.}

\smallskip
\noindent
If this happens, then we conclude that the graph $G$ remains acyclic even after the insertion of the edge $(u, v)$. 
We now justify this conclusion. Suppose that the insertion of the edge $(u,v)$ creates a cycle. Such a cycle defines a path $P_{v, u}$ from $v$ to $u$. Below, we make a  claim that will be proved later on.

\begin{claim}
\label{cl:path}
The path $P_{v, u}$ contains at least one node $x$ from the set $F_a$.
\end{claim}

\noindent Armed with Claim~\ref{cl:path}, we consider any node $x'$ in the path $P_{v,u}$ that belongs to the set $F_a$. Let $y' = \arg \min_{y \in B_d} \{ k(y) \}$.  Note that $k(y') = \min_{y \in B_a} k(y) \prec \min_{x \in F_a} k(x) \preceq  k(x')$. In particular, we infer that $k(y') \prec k(x')$. As  $y' \in B_d$, the node $u$ is reachable from  $y'$ (see Property~\ref{prop:forward:back}). Similarly, as the node $x'$ lies on the path $P_{v, u}$, the node $u$ is also reachable from $x'$. Since the node $u$ is reachable from both  the nodes $y' \in B_d$ and $x'$,  and since $k(y') \prec k(x)$, Corollary~\ref{cor:order:backward} implies that $x' \in B_d$. This leads to a contradiction, for $x' \in F_a$ and $F_a \cap B_d = \emptyset$ (see Property~\ref{prop:forward:back}). Hence, our initial assumption was wrong, and  the insertion of the edge $(u,v)$ does not create a cycle in $G$. It now remains to prove Claim~\ref{cl:path}.

\smallskip
\noindent {\em Proof of Claim~\ref{cl:path}.}
Applying the same argument used to justify condition (C1), we first observation that $v \in F_a \cup F_d$ and $u \in B_a \cup B_d$. As the subsets $F_a, F_d, B_a$ and $B_d$ are pairwise mutually exclusive (see Property~\ref{prop:forward:back}), we have $u \notin F_a \cup F_d$. Note that  if $v \in F_a$, then there is nothing further to prove.  Accordingly, for the rest of the proof we consider the scenario where $v \in F_d$. Since $v \in F_d$ and $u \notin F_d$, there has to be at least one node in the path $P_{v,u}$ that does not belong to the set $F_d$.  Let $x$ be the first such node, and let $y$ be the node that appears just before $x$ in the path $P_{v, u}$. Thus, we have $y \in F_d$, $x \notin F_d$ and $(y, x) \in E_{i,j}$. Hence, Corollary~\ref{cor:alive:dead} implies that $x \in F_a$. So the path $P_{v, u}$ contains some node from the set $F_a$.
\qed

\smallskip
\noindent
{\bf (C6)} {\em  $\max_{y \in B_a} k(y) \prec \max_{x \in F_d} k(x)$.}

\smallskip
\noindent
Similar to condition (C5), here we conclude that the graph $G$ remains acyclic.

\medskip
We now state an important corollary that follows from our stopping  conditions (C5) and (C6). It states that every node $x \in F_d$ appears before every node $y \in B_d$ in the total order $\prec$ in phase III.

\begin{corollary}
\label{cor:order:terminate}
We always have $\max_{x \in F_d} \{k(x)\} \prec \min_{y \in B_d} \{k(y)\}$.
\end{corollary}

\begin{proof}
Suppose that the corollary is false. Note that initially when the subroutine SEARCH($u, v$) begins execution, we have $F_d = B_d = \emptyset$ and hence the corollary is vacuously true at that time. Consider the first time-instant (say) $t$ when the corollary becomes false. Accordingly, we have:
\begin{equation}
\label{eq:time}
\max_{x \in F_d} \{k(x)\} \prec \min_{y \in B_d} \{k(y)\} \text{ just before time } t.
\end{equation}
One the following two events must have occurred at time $t$ for the corollary to get violated. 

\smallskip
\noindent (1) A  node $x' \in F_a$  was explored during a call to the subroutine EXPLORE-FORWARD($x'$). The subroutine EXPLORE-FORWARD($x'$) then moved the node $x'$ from the set $F_a$ to the set $F_d$, which violated the corollary. 
Note that a call to  EXPLORE-FORWARD(.) can only be made if  $k(x') \prec \min_{y \in B_d} \{k(y)\}$ just before time $t$ (see stopping condition (C5) and Property~\ref{prop:order}).
Thus, from~(\ref{eq:time}) we conclude that the corollary  remains satisfied even after adding the node $x'$ to the set $F_d$. This leads to a contradiction.

\smallskip
\noindent (2) A  node $x' \in F_a$   was explored during a call to  EXPLORE-FORWARD($x'$). The subroutine EXPLORE-FORWARD($x'$) then moved the node $x'$ from the set $F_a$ to the set $F_d$ 
, which violated the corollary. Applying an argument analogous to the one applied in case (1), we again reach a contradiction.
\end{proof}

The proof of  Lemma~\ref{lm:phaseIII:correct} follows immediately from the preceding discussion. Next, Lemma~\ref{lm:search:time} bounds the time spent in any single call to the subroutine SEARCH($u, v$).

\begin{lemma}
\label{lm:phaseIII:correct}
The subroutine SEARCH($u,v$) in Figure~\ref{fig:search} returns YES if the insertion of the edge $(u,v)$ creates a cycle in the graph $G$, and NO otherwise.
\end{lemma}

\begin{lemma}
\label{lm:search:time}
Consider any call to the subroutine SEARCH($u,v$). The time spent on this call is at most $\tilde{O}(m/n)$ times the size of the set $F_d$ at the end of the call.
\end{lemma}

\begin{proof}(Sketch)
Each call to  EXPLORE-FORWARD($x$) or EXPLORE-BACKWARD($x$) takes time proportional to the out-degree (resp. in-degree) of $x$ in the subgraph $G_{i,j}$. Under Assumption~\ref{assume}, the maximum in-degree and maximum out-degree of a node in $G_{i,j}$ are both at most $O(m/n)$. Thus, a single call to  EXPLORE-FORWARD($x$) or EXPLORE-BACKWARD($x$) takes $O(m/n)$ time. 

According to Property~\ref{prop:order}, whenever we want to explore a node during forward-search (resp. backward-search), we select a forward-alive (resp. backward-alive) node with minimum (resp. maximum) priority. This step can be implemented using a priority queue data structure in $\tilde{O}(1)$ time.

So the time spent by procedure SEARCH($u,v$) is at most $\tilde{O}(m/n)$ times the number of calls to the subroutines EXPLORE-FORWAD(.) or EXPLORE-BACKWARD(.). Furthermore, after each call to the subroutine EXPLORE-FORWAD(.) or EXPLORE-BACKWARD(.), the size of the set $F_d$ or $B_d$ respectively increases by one. 
Accordingly, the time spent on one call to  SEARCH($u,v$) is at most $\tilde{O}(m/n)$ times the size of the set $F_d \cup B_d$ at the end of the call. The lemma now follows from Property~\ref{prop:balance}.
\end{proof}

\medskip
\noindent {\bf Total time spent in phase III.} We now analyze the total time spent  in phase III, over the entire sequence of edge insertions in $G$.
For $l \in [1, m]$, consider the $l^{th}$ edge-insertion in the graph $G$, and let $t_l$ denote the size of the set $F_d$ at the end of phase III while handling this $l^{th}$-edge insertion. Lemma~\ref{lm:search:time} implies that the total time spent in phase III  is at most $\tilde{O}\left((m/n) \cdot \sum_{l=1}^m t_l\right)$. We now focus on upper bounding the sum $\sum_{l=1}^m t_l$. 

\begin{lemma}
\label{lm:search:time:total}
We have $\sum_{l=1}^m t_l^2 = O(n \tau)$. 
\end{lemma}

\begin{proof}
For any $l \in [1, m]$, let $F_d^{(l)}$ and $B_d^{(l)}$ respectively denote the sets $F_d$ and $B_d$ at the end of phase III while handling the $l^{th}$ edge-insertion in $G$. Furthermore, let $G^{(l)}$ and $G_{i,j}^{(l)}$ respectively denote the input graph $G$ and the subgraph $G_{i,j}$  after the $l^{th}$ edge-insertion in $G$. 

Suppose that the edge $(u,v)$ is the $l^{th}$ edge to be inserted into $G$. We focus on the procedure for handling this edge insertion. During this procedure,  if we find  $k(u) \prec k(v)$ in the beginning of phase III, then our algorithm immediately declares that the insertion of the edge $(u,v)$ does not create a cycle and moves on to phase IV. In such a scenario, we clearly have $F_d^{(l)} = B_d^{(l)} = \emptyset$ and hence $t_l = 0$. Accordingly, from now on we assume that $k(v) \prec k(u)$ in the beginning of phase III. Consider any two nodes $x \in F_d^{(l)}$ and $y \in B_d^{(l)}$.  The nodes $x$ and $y$ belong to the same subgraph $G_{i,j}^{(l)}$. Property~\ref{prop:forward:back} guarantees that there is a path $P_{y, x}$ from $y$ to $x$ in $G_{i,j}^{(l)}$ -- we can go from $y$ to $u$, take the edge $(u,v)$ and then go from $v$ to $x$. Hence, by Lemma~\ref{lm:propofpartition},  the ordered pair $(y, x)$ is $\tau$-related in $G^{(l)}$ with high probability. We condition on this event for the rest of the proof.  We now claim that there was no path from $y$ to $x$ in  $G^{(l-1)}$: this is the graph $G$ just before the $l^{th}$ edge-insertion, or equivalently, just after the $(l-1)^{th}$ edge-insertion. To see why this claim is true, we recall Property~\ref{lm:phase:order:correct}. This property states that in the beginning of phase III (after the $l^{th}$ edge-insertion) the total order $\prec$ on the node-set $V$ is a topological order in the graph $G^{(l-1)}$. Since $y \in B_d^{(l)}$ and $x \in F_d^{(l)}$, Corollary~\ref{cor:order:terminate} implies that $x$ appears before $y$ in the total order $\prec$ in phase III (after the $l^{th}$ edge-insertion). From these last two observations, we conclude that there is no path from $y$ to $x$ in $G^{(l-1)}$. As edges only get inserted into $G$ with the passage of time, this also implies that there is no path from $y$ to $x$ in the graph $G^{(l')}$, for all $l' < l$. Accordingly, the ordered pair $(y,x)$ is {\em not} $\tau$-related in the graph $G^{(l')}$ for any $l' < l$.

To summarize, for every node $x \in F_d^{(l)}$ and every node $y \in B_d^{(l)}$ the following conditions hold. (1) The ordered pair $(y,x)$ is $\tau$-related in the graph $G^{(l)}$. (2) For all $l' < l$, the ordered pair $(y,x)$ is {\em not} $\tau$-related in the graph $G^{(l')}$. Let $C$ denote a counter which keeps track of the number of sometime $\tau$-related pairs of nodes (see Definition~\ref{def:related:nodes:time}). Conditions (1) and (2) imply that every ordered pair of nodes $(y, x)$, where $y \in B_d^{(l)}$ and $x \in F_d^{(l)}$, contributes one towards the counter $C$.  A simple counting argument gives us: 
\begin{equation}
\label{eq:count:new}
\sum_{l=1}^m \left| F_d^{(l)} \right| \cdot \left| B_d^{(l)} \right| \leq C = O(n \tau)
\end{equation}
In the above derivation,  the last equality follows from Theorem~\ref{th:bound:related:nodes}. We now recall Property~\ref{prop:balance}, which says that our algorithm in phase III explores (almost) the same number of forward and backward nodes. In particular, we have $\left| F_d^{(l)} \right| \cdot \left| B_d^{(l)} \right| = O\left(\left| F_d^{(l)} \right|^2\right) = O(t_l^2)$ for all $l \in [1, m]$. This observation, along with~(\ref{eq:count:new}), implies that $\sum_{l=1}^m t_l^2 = O(n \tau)$. This concludes the proof of the lemma.
\end{proof}

\begin{corollary}
\label{cor:search:time:total}
We have $\sum_{l=1}^m t_l = O(\sqrt{mn\tau})$.
\end{corollary}

\begin{proof}
We partition the set of indices $\{1, \ldots, m\}$ into two subsets: 
$$X = \left\{ l \in [1, m] : t_l \leq \sqrt{n \tau/m} \right\} \text{ and } Y = \left\{ l \in [1, m] : t_l > \sqrt{n \tau/m}\right\}.$$ 
It is easy to check that $\sum_{l \in X} t_l \leq |X| \cdot \sqrt{n \tau/m} \leq m \cdot \sqrt{n \tau/m} = \sqrt{mn\tau}$.  Accordingly, for the rest of the proof we focus on bounding the sum $\sum_{l \in Y} t_l$. Towards this end, for each $l \in Y$, we first express the quantity $t_l$ as $t_l = \sqrt{n \tau/m} + \delta_l$, where $\delta_l > 0$. Now, Lemma~\ref{lm:search:time:total} implies that:
\begin{equation}
\label{eq:last}
\sum_{l \in Y} t_l^2 = \sum_{l \in Y} \left(\sqrt{n \tau/m} + \delta_l\right)^2  = O(n \tau)
\end{equation}
We also note that:
\begin{equation}
\label{eq:last:1}
 \sum_{l \in Y} \left(\sqrt{n \tau/m} + \delta_l\right)^2  \geq \sum_{l \in Y} \left(\delta_l \cdot \sqrt{n \tau/m}\right) =  \sqrt{n \tau/m} \cdot \sum_{l \in Y} \delta_l
\end{equation}
From~(\ref{eq:last}) and~(\ref{eq:last:1}), we get $\sqrt{n \tau/m} \cdot \sum_{l \in Y} \delta_l = O(n \tau)$, which in turn gives us: $\sum_{l \in Y} \delta_l = O\left(\sqrt{mn\tau}\right)$. This leads to the following upper bound on the sum $\sum_{l \in Y} t_l$.
\begin{equation}
\nonumber
\sum_{l \in Y} t_l = \sum_{l \in Y} \left(\sqrt{n \tau/m} + \delta_l\right) = \sum_{l \in Y} \sqrt{n \tau/m} + \sum_{l \in Y} \delta_l \leq m \cdot  \sqrt{n \tau/m} + O\left(\sqrt{mn\tau}\right)
= O\left(\sqrt{mn\tau}\right).
\end{equation}
This concludes the proof of the corollary.
\end{proof}

We are now ready to upper bound the total time spent by our algorithm in phase III.

\begin{lemma}
\label{lm:search:final:total:time}
We spend $\tilde{O}\left(\sqrt{m^3 \tau/n}\right)$ total time in phase III, over the entire sequence of edge-insertions.
\end{lemma}

\begin{proof}
 Lemma~\ref{lm:search:time} implies that  the total time spent in phase III is  $O\left( (m/n) \cdot \sum_{l=1}^m t_l \right)$. The lemma now follows from Corollary~\ref{cor:search:time:total}.
\end{proof}

\subsubsection{Phase IV: Ensuring that $\prec$ is a  topological ordering for $G^+$ (only when $G^+$ is acyclic)}
\label{sec:phase:final}

As in Section~\ref{sec:phase:cycle}, we let $G^-$ and $G^+$ respectively denote the graph $G$ just before and after the insertion of the edge $(u,v)$. If in phase III we detect a cycle, then we do not need to perform any nontrivial computation from this point onward, for the graph $G$ will contain a cycle after every future edge-insertion. Hence, throughout this section we assume that no cycle was detected in phase III, and as per Lemma~\ref{lm:phaseIII:correct} the graph $G^+$ is acyclic.  Our goal in phase IV is to update the total order $\prec$ so that it becomes a topological ordering in $G^+$. Towards this end, note that  $\prec$ does not change during phase III. Furthermore, if $k(u) \prec k(v)$ in phase III, then  the first three paragraphs of Section~\ref{sec:phase:cycle} imply that $\prec$ is already a topological ordering of $G^+$, and nothing  further needs to be done. Thus, from now on we assume that $k(v) \prec k(u)$ and $V(u) = V(v) = V_{i,j}$ for some $i, j \in [0, |S|]$ in phase III.

Recall the six terminating conditions for the subroutine SEARCH($u, v$) used in phase III (see the discussion after Corollary~\ref{cor:order:backward}). We have already assumed that we do not detect any cycle in phase III. Hence, the subroutine SEARCH($u, v$) terminates under one of the following four conditions: (C1), (C2), (C5) and (C6). How we update the total order $\prec$ in phase IV depends on the terminating condition under which the subroutine SEARCH($u,v$) returned in phase III. In particular, there are two cases to consider.

\medskip
\noindent {\bf Case 1. The subroutine SEARCH($u,v$) returned under condition (C2) or (C6) in phase III.} 

\smallskip
\noindent In this scenario, we update the total order $\prec$ by calling the subroutine described in Figure~\ref{fig:update:forward} (see Section~\ref{sec:codes}). In this subroutine, the symbols $F_d$ and $B_d$ respectively denotes the set of forward-dead and backward-dead nodes at the end of phase III. Similarly, we will use the symbols $F_a$ and $B_a$ respectively to denote the set of forward-alive and backward-alive nodes at the end of phase III. The subroutine works as follows.

When the subroutine SEARCH($u,v$) begins execution in phase III, we had $v \in F_a$ and $u \in B_a$. Since SEARCH($u,v$) returned under conditions (C2) or (C6), Property~\ref{prop:node:move} implies that $v \in F_d$ and $u \in B_d$ at the end of phase III.  Thus, when phase IV begins, let $v, x_1, \ldots, x_f$ be the nodes in $F_d$ in increasing order of priorities, so that $k(v) \prec k(x_1) \prec \cdots \prec k(x_f)$. Similarly, let $y_1, \cdots, y_b, u$ be the nodes in $B_d$ in increasing order of priorities, so that $k(y_1) \prec \cdots \prec k(y_b) \prec k(u)$.  By Corollary~\ref{cor:order:terminate}, we  have $k(x_f) \prec k(y_1)$. Now that the edge $(u,v)$ has been inserted, we need to update the relative ordering among the nodes in $F_d \cup B_d$.

Steps 1-8 in Figure~\ref{fig:update:forward} update the total order $\prec$ in such a way  that it satisfies the following properties. (1) We still have $k(v) \prec k(x_1) \prec \cdots \prec k(x_f)$. So the relative ordering among the nodes in $F_d$ does not change. (2) Consider any two nodes $x, y \in V$
 such that  $k(x) \prec k(x_f) \prec k(y)$ at the end of phase III. Then we still have $k(x) \prec k(x_f) \prec k(y)$ at the end of step 8 in Figure~\ref{fig:update:forward}. So the relative position of $x_f$ among all the nodes in $V$ does not change. (3)  The nodes in $F_d$ occur in consecutive positions in the total order $\prec$. Thus, at the end of step 8 it cannot be the case that $k(x') \prec k(x) \prec k(x'')$ if $x \notin F_d$ and $x', x'' \in F_d$. 
 
 \begin{claim}
 \label{cl:critical}
 Consider any edge $(x, y)$ in $G^-_{i,j}$ where $y \in B_d$ and $x \notin B_d$. Then  $k(x) \prec k(v)$ at the end of step 8 in Figure~\ref{fig:update:forward}.
 \end{claim}
 
 \begin{proof}
Since $y \in B_d$, $x \notin B_d$ and there is an edge from $x$ to $y$, Corollary~\ref{cor:alive:dead} implies that $x \in B_a$. Hence, the subroutine SEARCH($u,v$) returned under condition (C6), and {\em not} under condition (C2). By condition (C6), we have $k(x) \prec k(x_f)$ at the end of phase III. Since steps 1-8 in Figure~\ref{fig:update:forward} ensure that the nodes in $F_d$ occur in consecutive positions in $\prec$ and they do not change the relative position of $x_f$ among all the nodes in $V$, we get $k(x) \prec k(v)$ at the end of step 8 in Figure~\ref{fig:update:forward}.
 \end{proof}
 
 Steps 9-15 in Figure~\ref{fig:update:forward} further update the total order $\prec$ in such a way  that it satisfies the following properties. (4) We still have $k(y_1) \prec \cdots \prec k(y_b) \prec k(u)$. In words, the relative ordering among the nodes in $B_d$ does not change. (5) The node $u$ is placed immediately before the node $v$ in the total order $\prec$. This is consistent with the fact that the edge $(u,v)$ has been inserted into the graph $G$. (6) The nodes in $B_d$ occur in consecutive positions in the total order $\prec$. In other words, at the end of step 15 we cannot find any node $y \notin B_d$ and any two nodes $y', y'' \in F_d$ such that $k(y') \prec k(y) \prec k(y'')$. 
 
 To summarize, at this point in time, in the total order  $\prec$ the nodes $y_1, \ldots, y_b, u, v, x_1, \ldots, x_f$ occur consecutive to one another, and in this order. Accordingly, Corollary~\ref{cor:order:forward}, Corollary~\ref{cor:order:backward} and Claim~\ref{cl:critical} ensure that the total order $\prec$ remains a topological order in $G^-$ at this point in time. Since $u$ appears before $v$ in $\prec$, we also conclude that at this point in time $\prec$ is also a topological order in $G^+$.
 
 \medskip
\noindent {\bf Case 2. The subroutine SEARCH($u,v$) returned under condition (C1) or (C5) in phase III.} 

 \smallskip
 \noindent This case is completely analogous to case 1 above, and we omit its description.

\begin{lemma}
\label{lm:time:update:order}
We spend $\tilde{O}\left(  \sqrt{mn \tau} \right)$ time in phase IV, over the entire sequence of edge-insertions in $G$.
\end{lemma}

\begin{proof}(Sketch)
Steps 7 and 14 in Figure~\ref{fig:update:forward} can be implemented in $O(1)$ time using the ordered list data structure~\cite{DietzS87}. Hence, the time spent in phase IV after a given edge-insertion is proportional to the sizes of the sets $F_d$ and $B_d$ at the end of phase III, and  by Property~\ref{prop:balance}, the sizes of the sets $F_d$ and $B_d$ are (almost) equal to one another. For $l \in [1, m]$, let $t_l$ denote the size of the set $F_d$ at the end of phase III while handling the $l^{th}$ edge-insertion in $G$. We conclude that the total time spent in phase IV, during the entire sequence of edge-insertions in $G$, is given by $O\left(\sum_{l=1}^m t_l \right)$. The lemma now follows from Corollary~\ref{cor:search:time:total}.
\end{proof}

\subsection{Bounding the Total Update Time of Our Algorithm}
\label{sec:time}

We simply add up the total time spent by our algorithm in each of these four phases, throughout the entire sequence of edge-insertions in $G$. In particular, we invoke Lemma~\ref{lm:partition:time}, Lemma~\ref{lm:phase:order:time}, Lemma~\ref{lm:search:final:total:time} and Lemma~\ref{lm:time:update:order} and conclude that the total expected update time of our algorithm is at most:
\begin{equation}
\label{eq:total:time:final}
\tilde{O}\left( mn/\tau + n^2/\tau + \sqrt{m^3\tau/n} + \sqrt{mn\tau} \right) = \tilde{O}\left( mn/\tau +  \sqrt{m^3\tau/n}  \right).
\end{equation}
In the above derivation, we have made the assumption that $m = \Omega(n)$. Now, setting $\tau = n/m^{1/3}$, we get a total expected update time of $\tilde{O}(m^{4/3})$. This concludes the proof of Theorem~\ref{th:main}.

\newpage

\subsection{Pseudocodes for the subroutines in Phase III and Phase IV}
\label{sec:codes}

\begin{figure}[h!]
	\centerline{\framebox{
			\begin{minipage}{5.5in}
				\begin{tabbing}
					01. \  \ \=   $Q = F_d$. \\
					02. \> $x^* = \arg \max_{x \in Q} \{ k(x) \}$	 \\				
					03. \> $Q = Q \setminus \{x^*\}$. \\
					04. \> {\sc While} $Q \neq \emptyset$: \\
					05. \> \ \ \ \ \ \ \ \ \= $x' = \arg \max_{x \in Q} \{k(x)\}$. \\
					06. \> \> $Q = Q \setminus \{x'\}$. \\
					07. \> \> INSERT-BEFORE($x', x^*$). \\
					08. \> \> $x^* = x'$. \\
					09. \> $y^* = v$. \\
					10. \> $Q = B_d$. \\
					11. \> {\sc While} $Q \neq \emptyset$: \\
					12. \> \> $y' = \arg \max_{y \in Q} \{ k(y) \}$. \\
					13. \> \> $Q = Q \setminus \{ y' \}$. \\
					14. \> \> INSERT-BEFORE($y', y^*$). \\
					15. \> \> $y^* = y'$. 
					\end{tabbing}
			\end{minipage}
	}}
	\caption{\label{fig:update:forward} Subroutine:  UPDATE-FORWARD$(.)$ used in phase IV.}
\end{figure}

\begin{figure}[h!]
	\centerline{\framebox{
			\begin{minipage}{5.5in}
				\begin{tabbing}
					01. \   \= {\sc Initialize:} $F_a = \{v\}$,  $B_a = \{u\}$, $F_d = \emptyset$  and  $B_d = \emptyset$. \\
					02. \> {\sc While} $F_a \neq \emptyset$ AND $B_a \neq \emptyset$: \\
					03. \>  \ \ \ \ \ \ \  \= $x = \arg \min_{x' \in F_a} \{ k(x) \}$. \\
					04. \> \> {\sc If} $k(x) \succ \min_{y' \in B_d} \{ k(y') \}$, {\sc Then} \\
					05. \> \> \ \ \ \ \ \ \ \ \= {\sc Return} NO. \qquad // Insertion of the edge $(u,v)$ does not create a cycle. \\
					06. \> \> {\sc Else} \\
					07. \> \> \> EXPLORE-FORWARD($x$). \\
					09. \>  \> $y = \arg \max_{y' \in B_a} \{ k(y') \}$. \\
					10. \> \> {\sc If} $k(y) \succ \max_{x' \in B_d} \{ k(x') \}$, {\sc Then} \\
					11. \> \> \> {\sc Return} NO. \qquad // Insertion of the edge $(u,v)$ does not create a cycle. \\
					12. \> \> {\sc Else} \\
					13. \> \> \> EXPLORE-BACKWARD($y$). \\
					14. \> {\sc Return} NO. \qquad \qquad  \qquad // Insertion of the edge $(u,v)$ does not create a cycle.
				\end{tabbing}
			\end{minipage}
		}}
		\caption{\label{fig:search} Subroutine:  SEARCH$(u, v)$ used in phase III.}
	\end{figure}
	
	\begin{figure}[h!]
		\centerline{\framebox{
				\begin{minipage}{5.5in}
					\begin{tabbing}
						1. \   \=  $F_a = F_a \setminus \{x\}$ and  $F_d = F_d \cup \{x\}$. \\
						2. \> {\sc For all} $(x,x') \in E$ with $V(x) = V(x')$: \\
						3. \> \ \ \ \ \ \ \ \= {\sc If} $x' \in B_a \cup B_d$, {\sc Then} \\
						4. \> \> \ \ \ \ \ \ \ \ \  \= {\sc Retrun} YES. \qquad // Insertion of the edge $(u,v)$ creates a cycle. \\
						5. \> \> {\sc Else if} $x' \notin F_a \cup F_d$, {\sc Then} \\
						6. \> \> \> $F_a = F_a \cup \{x'\}$. 
					\end{tabbing}
				\end{minipage}
			}}
			\caption{\label{fig:forward} Subroutine:  EXPLORE-FORWARD$(x)$ used in phase III.}
		\end{figure}
	
	\begin{figure}[h!]
			\centerline{\framebox{
					\begin{minipage}{5.5in}
						\begin{tabbing}
							1. \   \=  $B_a = B_a \setminus \{y\}$ and  $B_d = B_d \cup \{y\}$. \\
							2. \> {\sc For all} $(y', y) \in E$ with $V(y') = V(y)$: \\
							3. \> \ \ \ \ \ \ \ \= {\sc If} $y' \in F_a \cup F_d$, {\sc Then} \\
							4. \> \> \ \ \ \ \ \ \ \ \  \= {\sc Retrun} YES. \qquad // Insertion of the edge $(u,v)$ creates a cycle. \\
							5. \> \> {\sc Else if} $y' \notin B_a \cup B_d$, {\sc Then} \\
							6. \> \> \> $B_a = B_a \cup \{y'\}$. 
						\end{tabbing}
					\end{minipage}
				}}
				\caption{\label{fig:backward} Subroutine:  EXPLORE-BACKWARD$(y)$ used in phase III.}
			\end{figure}

\printbibliography[heading=bibintoc] 

\appendix

\section{Full Version of Our Algorithm}
\label{app:sec:full}

 This section is organized as follows. In Section~\ref{app:sec:prelim}, we define some preliminary concepts and notations. In Section~\ref{app:sec:algo}, we  present our algorithm for incremental cycle detection and prove its correctness. In Section~\ref{app:sec:time}, we analyze the total update time of the algorithm, which leads to a proof of Theorem~\ref{th:main}.

\subsection{Preliminaries}
\label{app:sec:prelim}

Throughout the paper, we assume that the maximum degree of a node in $G$ is at most $O(1)$ times the average degree. It was observed by Bernstein and Chechik~\cite{BernsteinC18} that this assumption is without any loss of generality. 
\begin{assumption}~\cite{BernsteinC18}
\label{app:assume}
Every node in $G$ has an out-degree of  $O(m/n)$ and an in-degree of  $O(m/n)$.
\end{assumption}

We say that a node $u \in V$ is an {\em ancestor} of another node $v \in V$ iff there is a directed path from $u$ to $v$ in $G$. We let $A(v) \subseteq V$ denote the set of all ancestors of $v \in V$. Similarly, we say that $u$ is a {\em descendant} of $v$ iff there is a directed path from $v$ to $u$ in $G$. We let $D(v) \subseteq V$ denote the set of all descendants of $v$. A node is both an ancestor and a descendant of itself, that is, we have $v \in A(v) \cap D(v)$. We also fix an integral parameter $\tau \in [1, n]$ whose exact value will be determined later on.

 We  recall a crucial definition from Bernstein and Chechik~\cite{BernsteinC18}. First, note that if there is a path from a node $u$ to another node $v$ in $G$, then  $A(u) \subseteq A(v)$ and $D(v) \subseteq D(v)$. Such a pair of nodes is said to be {\em $\tau$-related} iff the number of nodes in each of the sets $A(v) \setminus A(u)$ and $D(u) \setminus D(v)$ does not exceed $\tau$. 

\begin{definition}~\cite{BernsteinC18}
\label{app:def:related:nodes}
We say that an  ordered pair of nodes $(u, v)$ is {\em $\tau$-related} in the graph $G$ iff there is a path from $u$ to $v$ in $G$, and $|A(v) \setminus A(u)| \leq \tau$ and $|D(u) \setminus D(v)| \leq \tau$. We emphasize that for the ordered pair $(u, v)$ to be $\tau$-related, it is {\em not} necessary that there be an edge $(u,v) \in E$.
\end{definition}

If two nodes $u, v \in V$ are part of a cycle, then clearly $A(u) = A(v)$ and $D(u) = D(v)$. In this case, it follows that both the ordered pairs $(u,v)$ and $(v, u)$ are $\tau$-related. In other words, if an ordered pair $(u, v)$ is not $\tau$-related, then there is no cycle containing both $u$ and $v$. Intuitively, therefore, the notion of $\tau$-relatedness serves as a {\em relaxation} of the notion of two nodes being part of a cycle. This is also the reason why this notion turns out to be extremely useful in designing an algorithm for incremental cycle detection.

Note that the graph $G$ keeps changing as more and more edges are inserted into it. Thus, it might  be the case that an ordered pair of nodes $(u, v)$ is {\em not} $\tau$-related in $G$ at some point in time, but {\em is} $\tau$-related in $G$ at some other point in time.  The following definition becomes relevant in light of this observation.

\begin{definition}~\cite{BernsteinC18}
\label{app:def:related:nodes:time}
We say that an {\em ordered} pair of nodes $(u, v)$ is {\em sometime $\tau$-related} in the graph $G$ iff it is $\tau$-related at some point in time during the sequence of edge insertions in $G$.
\end{definition}

In~\cite{BernsteinC18}, the following upper bound was derived on the number of sometime $\tau$-related pairs of nodes.

\begin{theorem}~\cite{BernsteinC18}
\label{app:th:bound:related:nodes}
The number of sometime $\tau$-related pairs of nodes in $G$ is at most $O(n \tau)$.
\end{theorem}

Following the framework developed in~\cite{BernsteinC18}, we will maintain a partition of the node-set $V$ into subsets $\{V_{i,j}\}$ and the subgraphs $\{G_{i,j} = (V_{i,j}, E_{i,j})\}$ induced by these subsets of nodes. At a high level, this partition serves as a useful proxy for determining if a given ordered pair of nodes is $\tau$-related.

We sample each node  $x \in V$ independently with probability $\log n/\tau$. Let $S \subseteq V$ denote the set of these sampled nodes.  The outcome of this random sampling defines a partition of the node-set $V$ into $(|S|+1)^2$ many subsets $\{V_{i,j}\}$, where $i,j \in [0, |S|]$. This is  formally define as follows.  For every node $v \in V$, let $A_S(v) = A(v) \cap S$ and $D_S(v) = D(v) \cap S$ respectively denote the set of ancestors and descendants of $v$ that have been sampled. Each subset $V_{i,j} \subseteq V$  is indexed by an ordered pair $(i, j)$ where $i \in [0, |S|]$ and $j \in [0, |S|]$.  A node $v \in V$ belongs to a subset $V_{i, j}$ iff $|A_S(v)| = i$ and $|D_S(v)| = j$. In words, the index $(i, j)$ of the subset $V_{i,j}$ specifies the number of sampled ancestors and sampled descendants each node $v \in V_{i,j}$ is allowed to have. It is easy to check that the subsets $\{V_{i,j}\}$ form a valid partition the node-set $V$. Let $E_{i, j} = \{ (u, v) \in E : u, v \in V_{i,j} \}$ denote the set of edges in $G$ whose both endpoints lie in $V_{i,j}$, and let $G_{i, j} = (V_{i,j}, E_{i,j})$ denote the subgraph of $G$ induced by the subset of nodes $V_{i,j}$.  We also define a total order $\prec^*$ on the subsets  $\{V_{i,j}\}$. For every two subsets $V_{i,j}$ and $V_{i',j'}$ in the partition, we have $V_{i,j} \prec^* V_{i',j'}$ iff either  $\{i < i'\}$ or $\{i = i' \text{ and } j > j'\}$.  We slightly abuse the notation by letting  $V(v)$ denote the unique subset $V_{i,j}$ which containing the  node $v \in V$. Consider any edge $(u, v) \in E$. If the two endpoints of the edge belong to two different subsets in the partition $\{V_{i,j}\}$, i.e., if $V(u) \neq V(v)$, then we refer to the edge $(u, v)$ as a {\em cross edge}. Otherwise, if $V(u) = V(v)$, then we refer to the edge $(u,v)$ as an {\em internal edge}.

We now state three lemmas that will be crucially used in our algorithm for incremental cycle detection. Although these lemmas were derived in~\cite{BernsteinC18}, for the sake of completeness we briefly describe their proofs here. Lemma~\ref{app:lm:cycle} states that the graph $G$ contains a cycle iff some subgraph $G_{i,j}$ contains a cycle. Hence, in order to detect a cycle in $G$ it suffices to only consider the edges that belong to the induced subgraphs $\{G_{i,j}\}$. Lemma~\ref{app:lm:cross}, on the the other hand, implies that if the graph $G$ is acyclic, then it admits a topological ordering $\prec$ that is {\em consistent} with the total order $\prec^*$, meaning   that $u \prec v$ for all $u, v \in V$ with $V(u) \prec^* V(v)$. Finally, Lemma~\ref{app:lm:partition:related:nodes} states that whenever a subgraph $G_{i,j}$ contains  a path from a node $u$ to some other node $v$, with high probability  the ordered pair $(u,v)$ is $\tau$-related in the input graph $G$.

\begin{lemma}~\cite{BernsteinC18}
\label{app:lm:cycle}
If there is a cycle in  $G = (V, E)$, then every edge of that cycle is an internal edge. 
\end{lemma}

\begin{proof}(Sketch)
The key observation is that if two nodes $u$ and $v$ lie on a cycle, then they have exactly the same set of ancestors and descendants, that is, $A(u) = A(v)$ and $D(u) = D(v)$. For such a pair of nodes $u$ and $v$, we clearly have $A_S(u) = A_S(v)$ and $D_S(u) = D_S(v)$. In other words, if there is an edge $(u, v) \in E$ that is part of a cycle, then both the endpoints of that edge belong to the same subset in the partition $\{V_{i,j}\}$, so that $V(u) = V(v)$. Hence, every edge that is part of a cycle is an internal edge.
\end{proof}

\begin{lemma}~\cite{BernsteinC18}
\label{app:lm:cross}
For every cross edge $(u, v) \in E$, we have $V(u) \prec^* V(v)$. 
\end{lemma}

\begin{proof}(Sketch)
Consider any cross edge $(u, v) \in E$, where $V(u) = V_{i,j}$ and $V(v) = V_{i', j'}$. Since $(u, v)$ is a cross edge, by definition $V_{i,j} \neq V_{i', j'}$. Clearly, every ancestor of $u$ is also an ancestor of $v$, and thus we have $A(u) \subseteq A(v)$. This implies that $A_S(u) \subseteq A_S(v)$ and hence $|A_S(u)| \leq |A_S(u)|$. Now, consider two possible cases. Either $|A_S(u)| < |A_S(v)|$ or $|A_S(u)| = |A_S(v)|$. In the former case, we  have $i = |A_S(u)| < i' = |A_S(v)|$, which means that $V_{i,j} \prec^* V_{i',j'}$.  In the latter case, we have $i = i'$. Here, we note that every descendant of $v$ is also a descendant of $u$, and using exactly the same argument as before we conclude that $|D_S(u)| \geq |D_S(v)|$, which gives us: $j = |D_S(u)| \geq j' = |D_S(v)|$.  Since $(u, v)$ is a cross edge, we have  $V_{i, j} \neq V_{i', j'}$. Furthermore,  we are now considering the case where $i = i'$. Hence, we cannot have $j = j'$, and accordingly, we derive that $j > j'$. As $i = i'$ and $j > j'$, we again get $V_{i, j} \prec^* V_{i', j'}$. 
\end{proof}

\begin{lemma}~\cite{BernsteinC18}
\label{app:lm:partition:related:nodes}
Consider any two nodes $u, v \in V_{i,j}$ for some $i, j \in [0, |S|]$. If there is a path from $u$ to $v$ in the subgraph $G_{i,j}$, then with high probability the ordered pair $(u, v)$ is $\tau$-related in $G$. 
\end{lemma}

\begin{proof}
Suppose that there is a path from $u$ to $v$ in the subgraph $G_{i,j}$, but  the ordered pair $(u,v)$ is {\em not} $\tau$-related in $G$. Then, either $|A(v) \setminus A(u)| > \tau$ or $|D(u) \setminus D(v)| > \tau$. For the rest of the proof, we assume that $|D(u) \setminus D(v)| > \tau$. An analogous argument applies in the other case. Each node $x \in D(u) \setminus D(v)$ is sampled in $S$ independently with probability $\log n/\tau$. Since $|D(u) \setminus D(v)| > \tau$, by linearity of expectation at least $\log n$ nodes from $D(u) \setminus D(v)$ are sampled in $S$. Applying Chernoff bound, we conclude that $S \cap \left(D(u) \setminus D(v)\right) \neq \emptyset$ with high probability. We condition on this event, and consider a node $x' \in S \cap \left(D(u) \setminus D(v)\right)$. Since $x' \in D(u)$, $x' \notin D(v)$ and $x' \in S$, we infer that $x' \in D_S(u)$ and $x' \notin D_S(v)$, and we get $D_S(u) \neq D_S(v)$. As there is a path from $u$ to $v$ in $G$, we have $D_S(v) \subseteq D_S(u)$. The last two observations, taken together, imply that $|D_S(u)| > |D_S(v)|$. Since both the nodes $u$ and $v$ belong to the same subset $V_{i,j}$, by definition we also have $|D_S(u)| = |D_S(v)|$. This leads to a contradiction and  our initial assumption, therefore, must have been wrong. This concludes the proof of the lemma.
\end{proof}

\subsection{The algorithm}
\label{app:sec:algo}

Since edges never get deleted from the graph $G$, our algorithm does not have to do anything once it detects a cycle (for the graph will continue to have a cycle after every edge-insertion in the future). Accordingly, we assume that the graph $G$ has  remained acyclic throughout the sequence of edge insertions till the present moment, and our goal is to check if  the next edge-insertion creates a cycle in $G$. Our algorithm maintains a topological ordering $\prec$ of the node-set $V$ in the graph $G$ that is {\em consistent} with the total order $\prec^*$ on the subsets of nodes $\{V_{i,j}\}$, as defined in Section~\ref{app:sec:algo}. Specifically, we maintain a {\em priority} $k(x)$ for every node $x \in V$, and for every two nodes $x, y \in V$ with $V(x) \prec^* V(y)$ we  ensure that $k(x) \prec k(y)$. As long as $G$ remains acyclic, the existence of such a topological ordering $\prec$ is guaranteed by Lemma~\ref{app:lm:cross}.

\medskip
\noindent {\bf Data Structures.} We maintain the partition $\{V_{i,j}\}$ of the node-set $V$ and the subgraphs $\{G_{i,j} = (V_{i,j}, E_{i,j})\}$ induced by the subsets in this partition. We  use an {\em ordered list} data structure~\cite{DietzS87} on the node-set $V$ to implicitly maintain the priorities $\{k(x)\}$ associated with the topological ordering $\prec$. This data structure supports each of the following operations in $O(1)$ time. 
\begin{itemize}
\item INSERT-BEFORE($x, y$): This inserts the node $y$ just before the node $x$ in the topological ordering.
\item INSERT-AFTER($x, y$): This inserts the node $y$ just after the node $x$ in the topological ordering. 
\item DELETE($x$): This deletes the node $x$ from the existing topological ordering.
\item COMPARE($x, y$): If $k(x) \prec k(y)$, then this returns YES, otherwise this returns NO.
\end{itemize}

\noindent The implementation of our algorithm requires the creation of   two {\em dummy nodes} $x_{i,j}$ and $y_{i,j}$ in every subset $V_{i, j}$. We ensure that $k(x_{i,j}) \prec k(x) \prec k(y_{i,j})$ for all $x \in  V_{i, j}$. In words, the dummy node $x_{i,j}$ (resp. $y_{i,j}$) comes {\em first} (resp. {\em last}) in the topological order among all the nodes in $V_{i,j}$. Further, for all nodes $x \in V$ with $V(x) \prec V_{i,j}$ we have $k(x) \prec k(x_{i,j})$, and for all nodes $x \in V$ with $V_{i,j} \prec V(x)$ we have $k(y_{i,j}) \prec k(x)$.

\medskip
\noindent {\bf Handling the insertion of an edge $(u,v)$ in $G$.} By induction hypothesis, suppose that the graph $G$ currently does not contain any cycle and we are maintaining the topological ordering $\prec$ in $G$. At this point, an edge $(u,v)$ gets inserted into $G$. Our task now is to first figure out if the insertion of this edge creates a cycle, and if not, then to update the topological ordering $\prec$. We perform this task in four {\em phases}, as described below.
\begin{enumerate}
\item In phase I, we update the subgraphs $\{ G_{i,j}\}$. See Section~\ref{app:sec:phase:partition} for details.
\item In phase II, we update the total order $\prec$ on the node-set $V$ to make it consistent with the total order $\prec^*$ on the subsets of nodes $\{V_{i,j}\}$. See Section~\ref{app:sec:phase:order} for details.
\item In phase III, we check if the edge-insertion  creates a cycle in $G$. See Section~\ref{app:sec:phase:cycle} for details.
\item If phase III fails to detect a cycle, then in phase IV we further update (if necessary) the total order $\prec$ so as to ensure that it is a topological order in the current graph $G$. See Section~\ref{app:sec:phase:final} for details.
\end{enumerate}

\subsubsection{Phase I: Updating the subgraphs $\{ G_{i,j} = (V_{i,j}, E_{i,j}) \}$}
\label{app:sec:phase:partition}

In this phase, we update the subgraphs $\{ G_{i,j}\}$. Lemma~\ref{app:lm:partition:time} upper bounds the total time spent in this phase. Although the lemma was derived in~\cite{BernsteinC18}, we include its proof  here for the sake of completeness.

\begin{lemma}~\cite{BernsteinC18}
\label{app:lm:partition:time}
There is an algorithm for maintaining the subgraphs $\{ G_{i,j} = (V_{i,j}, E_{i,j})\}$ in an incremental setting with $\tilde{O}(m n/\tau)$ expected total update time.
\end{lemma}

\begin{proof}(Sketch) We first show how to maintain the subsets of nodes $\{V_{i, j}\}$. 
The key observation is that in the incremental setting, single-source reachability can be maintained in $O(m)$ total update  time. Specifically, for any  node $v \in V$, we can maintain the sets $A(v)$ and $D(v)$ in $O(m)$ total update time. We use this  subroutine as follows. Initially, when the graph is empty, we set $A_S(v) = D_S(v) = \emptyset$ for all nodes $v \in V \setminus S$ and $A_S(v) = D_S(v) = \{v \}$ for all nodes $v \in S$. Subsequently, while processing the sequence of edge insertions into the graph, we  run $|S|$ incremental single source reachability algorithms -- one for each sampled node $s \in S$. Whenever one of these incremental algorithms (say for the sampled node $s \in S$) inserts a node $v \in V$ into the set $A(s)$ or the set $D(s)$, we respectively insert the node $s$ to the sets $A_S(v)$  or $D_S(v)$. This procedure correctly maintains the sets $A_S(v)$ and $D_S(v)$ for all  $v \in V$. Its  total update time is equal to $|S|$ times the total update time of one incremental single-source reachability algorithm. We accordingly get an expected total update time of $O(|S| \cdot m) = \tilde{O}(m n/\tau)$. The last equality holds since $\mathbf{E}[|S|] = n\log n/\tau$ by linearity of expectation. Now, note that it is straightforward to extend this procedure to maintain the subsets $\{V_{i,j}\}$ without incurring any overhead in the total update time:  We  keep two counters for each node $v \in V$, to keep track of the sizes of the sets $A_S(v)$ and $D_S(v)$. Whenever a sampled node $s$ gets added to one of these sets, we increment the corresponding counter and accordingly move the node $v$ from one subset to another in the partition $\{V_{i,j}\}$.

We now show how to maintain the subsets of edges $\{E_{i,j}\}$. Whenever a node $v$ moves from one subset (say) $V_{i,j}$ to another subset (say) $V_{i',j'}$, we visit all the (incoming ``and" outgoing) neighbors of $v$ and inform them about the fact that $v$ has moved to a new subset $V_{i',j'}$. While visiting a neighbor $u$ of $v$ we also check which subset does the node $u$ currently belong to, and this helps us identify the set of edges in $E_{i,j}$ that are incident on $v$. Overall, this takes time proportional to the sum of the in and out degrees of $v$. Under Assumption~\ref{app:assume},   the latter quantity  is  at most $O(m/n)$. Let $c_v$ be a counter which keeps track of the total number of times a node $v$ moves from one subset to another in the partition $\{ V_{i,j} \}$. The preceding discussion implies that the time spent on the node $v$ is at most $O(c_v \cdot m/n)$, and the total time spent on maintaining the subsets of edges $\{ E_{i,j} \}$ is at most $\sum_{v \in V} O(c_v \cdot m/n)$. Now, note that as edges only keep getting inserted in $G$ in the incremental setting, the sets $A_S(v)$ and $D_S(v)$ can only keep getting larger and larger with the passage of time. In particular,  whenever a node $v$ moves from one subset $V_{i,j}$ to another subset $V_{i', j'}$, we must have either $\{i' \geq i+1 \text{ and } j' \geq j\}$ or $\{ i' \geq i \text{ and } j' \geq j+1\}$. Since  $0 \leq i, j \leq |S|$ for every subset $V_{i,j}$, we infer that a node $v$ can move from one subset to another at most $2 \cdot |S|$ times, and hence we have $c_v \leq 2 \cdot |S|$ for all $v \in V$. Accordingly, the total time spent on maintaining the subsets of edges $\{E_{i,j}\}$ is upper bounded by $\sum_{v \in V} O(c_v \cdot (m/n)) = \sum_{v \in V} O(|S| \cdot (m/n)) = O(|S| \cdot m) = \tilde{O}(m n/\tau)$. The last equality holds since $n \log n/\tau$ nodes are sampled in $S$ in expectation.  

From the above discussion, it follows that we can maintain both the subsets of nodes $\{V_{i,j}\}$ and the subsets of edges $\{E_{i,j}\}$ in $\tilde{O}(mn/\tau)$ total expected update time. Hence, we can clearly maintain the induced subgraphs $\{G_{i,j}\}$ also in $\tilde{O}(mn/\tau)$ total expected update time. 
\end{proof}

\subsubsection{Phase II: Making $\prec$ consistent with the total order $\prec^*$ on the subsets of nodes $\{V_{i,j}\}$}
\label{app:sec:phase:order}

Let $G^-$ and $G^+$ respectively denote the graph $G$ just before and just after the insertion of the edge $(u,v)$. Similarly, for every node $x \in V$,  let $V^-(x)$ and $V^+(x)$ respectively denote the subset $V(x)$ just before and just after the insertion of the edge $(u, v)$.
Furthermore, let $X = \{ x \in V : V^-(x) \neq V^+(x)\}$ denote the set of all nodes that move from one subset to another during phase I in the partition $\{ V_{i,j} \}$.  Finally, for each $i, j \in [0, |S|]$, define  two subsets $X^{up}_{i,j} = \{ x \in X : V^-(x) \prec V^+(x) = V_{i,j} \}$ and $X^{down}_{i,j} = \{ x \in X : V_{i,j} = V^+(x) \prec V^-(x) \}$. In words, during phase I the nodes $x \in X^{up}_{i,j}$ move to the subset $V_{i,j}$  from a subset of lower priority, i.e., they move {\em up}. The nodes $x \in X^{down}_{i,j}$, on the other hand, move to the subset $V_{i,j}$ during phase I from a subset of higher priority, i.e., they move {\em down}. The set $X$ is clearly partitioned into the subsets $\{X_{i,j}^{up}\}$ and $\{X^{down}_{i,j}\}$. We need to change the positions of the nodes $x \in X$ in the total order $\prec$ to make it consistent with the total order $\prec^*$ defined over $\{V_{i,j}\}$. This is done as follows.
\begin{itemize}
\item For every nonempty subset $X^{up}_{i,j}$, we  change the positions of the nodes $x \in X^{up}_{i,j}$ in the total order $\prec$ by calling the subroutine  MOVE-UP($i,j$). See Figure~\ref{app:fig:move:up} for details. Intuitively, this subroutine ensures that the following three conditions are satisfied by the total order $\prec$. (1) All the nodes $x \in X^{up}_{i,j}$ are placed in between the dummy nodes $x_{i,j}$ and $y_{i,j}$. In other words, the nodes in $X^{up}_{i,j}$ are placed within the designated boundary of the interval that ought to be defined by the nodes in the subset $V_{i,j}$. (2) The relative ordering among the nodes within $X^{up}_{i,j}$ does not change, although each of the nodes in $X^{up}_{i,j}$ itself moves from one position to another in the total order  $\prec$. (3) When the subroutine returns, every node $x \in X^{up}_{i,j}$ is still placed {\em before} every non-dummy node $x' \in V$ with $V^-(x') = V^+(x') =  V_{i,j}$ in the total order $\prec$.
\item Analogously, for every nonempty subset $X^{down}_{i,j}$, we  change the positions of the nodes $y \in X^{down}_{i,j}$ in the total order  $\prec$ by calling the subroutine MOVE-DOWN($i,j$). See Figure~\ref{app:fig:move:down} for  details.  This subroutine ensures that the following three conditions are satisfied by the total order $\prec$. (1) All the nodes $y \in X^{down}_{i,j}$ are placed in between the dummy nodes $x_{i,j}$ and $y_{i,j}$. In other words, the nodes in $X^{down}_{i,j}$ are placed within the designated boundary of the interval that ought to be defined by the nodes in the subset $V_{i,j}$. (2) The relative ordering among the nodes within $X^{down}_{i,j}$ does not change, although each of the nodes in $X^{down}_{i,j}$ itself moves from one position to another in the total order  $\prec$. (3) When the subroutine returns, every node $y \in X^{down}_{i,j}$ is still placed {\em after} every non-dummy node $y' \in V$ with $V^-(y') = V^+(y') = V_{i,j}$ in the total order $\prec$.
\end{itemize}
From the description of the procedure above, we immediately get the following two corollaries.

\begin{corollary}
\label{app:cor:phase:order}
At the end of phase II the total order $\prec$ on $V$ is consistent with the total order $\prec^*$  on  $\{V_{i,j}\}$. Specifically, for any two nodes $x$ and $y$, if $V(x) \prec^* V(y)$, then we also have $k(x) \prec k(y)$. 
\end{corollary}

\begin{corollary}
\label{app:cor:phase:order:1}
Consider any two nodes $x$ and $y$ that belong to the same subset (say) $V_{i,j}$ at the end of phase I, that is, we have $V^+(x) = V^+(y) = V_{ij}$. Then, while updating the total order $\prec$ in phase II, our algorithm does not change the relative ordering among these two nodes $x$ and $y$.  In other words, we get $k(x) \prec k(y)$ at the end of phase II iff $k(x) \prec k(y)$ at the end of phase I.
\end{corollary}

\begin{figure}[h!]
	\centerline{\framebox{
			\begin{minipage}{5.5in}
				\begin{tabbing}
					1. \ \ \  \= $Q = X^{up}_{i,j}$. \\
					2. \> {\sc While} $Q \neq \emptyset$: \\
					3. \> \ \ \ \ \ \ \ \= $x = \arg \max_{x' \in Q} \{k(x')\}$. \\
					4. \> \> $Q \leftarrow Q \setminus \{x\}$. \\
					5. \>  \> INSERT-AFTER($x_{i,j}, x$).			
					\end{tabbing}
			\end{minipage}
	}}
	\caption{\label{app:fig:move:up} Subroutine:  MOVE-UP$(i, j)$.}
\end{figure}

\begin{figure}[h!]
	\centerline{\framebox{
			\begin{minipage}{5.5in}
				\begin{tabbing}
					1. \ \ \  \= $Q = X^{down}_{i,j}$. \\
					2. \> {\sc While} $Q \neq \emptyset$: \\
					3. \> \ \ \ \ \ \ \ \=  $x = \arg \min_{x' \in Q} \{k(x')\}$. \\
					4. \> \> $Q \leftarrow Q \setminus \{x\}$. \\
					5. \>  \> INSERT-BEFORE($y_{i,j}, x$).			
					\end{tabbing}
			\end{minipage}
	}}
	\caption{\label{app:fig:move:down} Subroutine:  MOVE-DOWN$(i, j)$.}
\end{figure}

We now prove a crucial lemma that will be useful when we go to phase III.

\begin{lemma}
\label{app:lm:phase:order:correct}
At the end of phase II the total order $\prec$ on $V$  remains a valid topological ordering of $G^-$, where $G^-$ denotes the graph $G$ just before the insertion of the edge $(u,v)$.
\end{lemma}

\begin{proof}(Sketch)
We show that $k(x) \prec k(y)$ at the end of phase II for every edge $(x, y)$ in $G^-$. We consider two cases. 
{\em Case (1). The two endpoints $x$ and $y$ belong to two different subsets at the end of phase I, that is,  $V^+(x) \neq V^+(y)$.} In this case,   Lemma~\ref{app:lm:cross} implies that $V^+(x) \prec^* V^+(y)$. Since the total order $\prec$ is consistent with $\prec^*$ at the end of phase II (see Corollary~\ref{app:cor:phase:order}), we  also have $k(x) \prec k(y)$ at this point in time. {\em Case (2). The two endpoints $x$ and $y$ belong to the same subset (say) $V_{i,j}$ at the end of phase II, that is, $V^+(x) = V^+(y) = V_{i,j}$.} Since $\prec$ was a topological ordering of $G^-$ before the insertion of the edge $(u,v)$, we have $k(x) \prec k(y)$ just before phase II.  Corollary~\ref{app:cor:phase:order:1} implies that the relative ordering between the two nodes $x$ and $y$ does not change during phase II. Hence, even at the end of phase II we have $k(x) \prec k(y)$. This concludes the proof of the lemma.
\end{proof}

We now bound the total time spent by the algorithm in phase II.

\begin{lemma}
\label{app:lm:phase:order:time}
The total time spent in phase II across all edge-insertions is at most $\tilde{O}(n^2/\tau)$.
\end{lemma}

\begin{proof}(Sketch)
The time taken to identify and construct the nonempty subsets $X_{i,j}^{up}$ and $X_{i,j}^{down}$ is subsumed by the time spent in phase I in updating the partition $\{V_{i,j} \}$. We now bound the total time spent across all the calls to the subroutines MOVE-UP($i,j$) and MOVE-DOWN($i,j$).

Step 3 in Figure~\ref{app:fig:move:up} (resp. Figure~\ref{app:fig:move:down}) can be implemented in $\tilde{O}(1)$ time by using a max-heap (resp. min-heap) data structure. Similarly, steps 5 in Figure~\ref{app:fig:move:up} and Figure~\ref{app:fig:move:down} can be implemented in $\tilde{O}(1)$ time by using the ordered-list data structure~\cite{DietzS87}. We therefore conclude that the time spent in a single call to the subroutine MOVE-UP($i,j$) or MOVE-DOWN($i,j$) is at most $\tilde{O}\left(\left|X^{up}_{i,j}\right|\right)$ or $\tilde{O}\left(\left|X^{down}_{i,j}\right|\right)$, respectively. From the definitions of the sets $X^{up}_{i,j}$ and $X^{down}_{i,j}$, it follows that the total time spent on the calls to these two subroutines during the entire sequence of edge-insertions in $G$ is at most $\tilde{O}(C)$, where $C$ is a counter that keeps track of the number of times some node moves from one subset in the partition $\{V_{i, j}\}$ to another.

It now remains to upper bound the value of the counter $C$. Towards this end, recall that  a node $x \in V$ belongs to a subset $V_{i,j}$ iff $|A_S(x)| = i$ and $|D_S(x)| = j$.  As more and more edges keep getting inserted in $G$, the node $x$ can never lose a sampled node in $S$ as its ancestor or descendent. Instead, both the sets $A_S(x)$ and $D_S(x)$ can only grow with the passage of time. In particular, each time the node $x$ moves from one subset in the partition $\{V_{i,j} \}$ to another, either $|A_S(v)|$ or $|D_S(v)|$ increases by at least one. Since $|A_S(v)|, |D_S(v)| \in [0, |S|]$, each node $x$ can move from one subset in the partition $\{V_{i,j}\}$ to another at most $2 \cdot |S|$ times. In other words, each node in $V$ contributes at most $2 \cdot |S|$ to the counter $C$, and  we have $C \leq |V| \cdot 2 |S| = O(n |S|)$. Since $\mathbf{E}[|S|] = \tilde{O}(n/\tau)$, the total expected time spent in all the calls to the subroutines MOVE-UP($i,j$) and MOVE-DOWN($i,j$) is at most $\tilde{O}(n^2/\tau)$, and the total time spent in phase II  is also at most $\tilde{O}(n^2/\tau)$. 
\end{proof}

\subsubsection{Phase III: Checking if the insertion of the edge $(u,v)$ creates a cycle.} 
\label{app:sec:phase:cycle}

As in Section~\ref{app:sec:phase:order}, let $G^-$ and $G^+$ respectively denote the graph $G$  before and after the insertion of the edge $(u, v)$. Consider the total order $\prec$  on the set of nodes $V$ in the beginning of phase III (or, equivalently, at the end of phase II). Corollary~\ref{app:cor:phase:order} guarantees that  $\prec$ is consistent with the total order $\prec^*$ on $\{V_{i,j}\}$.  Lemma~\ref{app:lm:phase:order:correct}, on the other hand, guarantees that $\prec$ is a valid topological ordering in $G^-$.  We will use these two properties of the total order $\prec$ throughout the current phase.

In phase III, our goal is to determine if the insertion of the edge $(u, v)$ creates a cycle in $G$. Note that if  $k(u) \prec k(v)$, then   $\prec$ is also a valid topological ordering in $G^+$ as per Lemma~\ref{app:lm:phase:order:correct}, and clearly the insertion of the edge $(u,v)$ does not create a cycle.  The difficult case occurs when $k(v) \prec k(u)$. In this case, we first infer that $V(u) = V(u)$. In words, both the nodes $u$ and $v$ belong to the same subset in the partition $\{V_{i,j}\}$ at the end of phase II. This is because of the following reason. The total order $\prec$ is consistent with the total order $\prec^*$ as per Corollary~\ref{app:cor:phase:order}. Accordingly, since $k(v) \prec k(u)$, we conclude that if $V(v) \neq V(u)$ then  $V(v) \prec^* V(u)$. But this would contradict Lemma~\ref{app:lm:cross} as there is a cross edge from $u$ to $v$.

To summarize, for the rest of this section we assume that $k(v) \prec k(u)$ and $V(v) = V(u) = V_{i,j}$ for some $i, j \in [0, |S|]$. We have to check if there is a path $P_{v,u}$ from $v$ to $u$ in $G^-$. Along with the edge $(u,v)$, such a path $P_{v,u}$ will define a cycle in $G^+$. Hence, by Lemma~\ref{app:lm:cycle}, every edge $e$ in such a path $P_{v, u}$ will belong to the subgraph $G_{i,j} = (V_{i,j}, E_{i,j})$. Thus, from now on our task is to determine if there is a path $P_{v,u}$ from $v$ to $u$ in $G_{i,j}$. We perform this task by calling the subroutine SEARCH($u,v$) described below. The pseudocode for this subroutine is given in Figure~\ref{app:fig:search}.

\medskip
\noindent {\bf SEARCH($u,v$).} 
We simultaneously conduct two searches in order to find the path $P_{v, u}$: A {\em forward search} from $v$, and a {\em backward search} from $u$. Specifically, let $F$ and $B$ respectively denote the set of nodes visited by the forward search and the backward search till now. Initially, we have $F = \{v\}$ and $B = \{u\}$, and we always ensure that $F \cap B = \emptyset$. A node in $F$ (resp. $B$) is referred to as a forward (resp. backward) node.   Every forward node $x \in F$ is reachable from the node $v$ in $G_{i,j}$. In contrast, the node $u$ is reachable from every backward node $x \in B$ in $G_{i,j}$. We further classify each of the sets $F$ and $B$ into two subsets: (1) $F_a \subseteq F$ and $F_d = F \setminus F_a$, (2) $B_a \subseteq B$ and $B_d = B \setminus B_a$. The nodes in $F_a$ and $B_a$ are called {\em alive}, whereas the nodes in $F_d$ and $B_d$ are called {\em dead}. Intuitively, the dead nodes have already been {\em explored} by the search, whereas the alive nodes have not yet been explored. 

\begin{property}
\label{app:prop:forward:back}
Every node $x \in F_a \cup F_d$ is reachable from the node $v$ in $G_{i,j}$, and the node $u$ is reachable from every node $y \in B_a \cup B_d$ in $G_{i,j}$. The sets $F_a, F_d, B_a$ and $B_d$ are pairwise mutually exclusive.
\end{property}

A simple strategy for exploring a  forward and alive node $x \in F_a$ is as follows. For each of its outgoing edges $(x, y) \in E_{i,j}$, we check if $y \in B$. If yes, then we have detected a path from $v$ to $u$: This path goes from $v$ to $x$ (this is possible since $x$ is a forward node), follows the edge $(x, y)$, and then from $y$ it goes to $u$ (this is possible since $y$ is a backward node).  Accordingly, we stop and report that the graph $G^+$ contains a cycle. In contrast, if $y \notin B$ and $y \notin F$, then we insert $y$ into the set $F_a$ (and $F$), so that $y$ becomes a forward and alive node  which will be explored in future. In the end, we move the node $x$ from the set $F_a$ to the set $F_d$. We refer to the subroutine that explores a node $x \in F_a$ as EXPLORE-FORWARD($x$). See Figure~\ref{app:fig:forward}.

Analogously, we explore a  backward and alive node $x \in B_a$ is as follows. For each of its incoming edges $(y, x) \in E_{i,j}$, we check if $y \in F$. If yes, then there is a path from $v$ to $u$: This path goes from $v$ to $y$ (this is possible since $y$ is a forward node), follows the edge $(y, x)$, and then from $x$ it goes to $u$ (this is possible since $x$ is a backward node).  Accordingly, we stop and report that the graph $G^+$ contains a cycle. In contrast, if $y \notin F$ and $y \notin B$, then we insert $y$ into the set $B_a$ (and $B$), so that $y$ becomes a backward and alive node  which will be explored in future. In the end, we move the node $x$ from the set $B_a$ to the set $B_d$. We refer to the subroutine that explores a node $x \in B_a$ as EXPLORE-BACKWARD($x$). See Figure~\ref{app:fig:backward}.

\begin{property}
\label{app:prop:node:move}
Once a node $x \in F_a$ (resp. $x \in B_a$) has been explored, we delete it from the set $F_a$ (resp. $B_a$) and insert it into the set $F_d$ (resp. $B_d$).
\end{property}

While exploring a node $x \in F_a$ (resp. $x \in B_a$), we ensure that all its outgoing (resp. incoming) neighbors are included in $F$ (resp. $B$). This leads to the following important corollary.
\begin{corollary}
\label{app:cor:alive:dead}
Consider any edge $(x, y) \in E_{i,j}$. At any point in time, if $x \in F_d$, then at that time we  also have $y \in F_a \cup F_d$.  Similarly, at any point in time, if $y \in B_d$, then at that time we also have $x \in B_a \cup B_d$.
\end{corollary}

\begin{figure}[h!]
	\centerline{\framebox{
			\begin{minipage}{5.5in}
				\begin{tabbing}
					01. \   \= {\sc Initialize:} $F_a = \{v\}$,  $B_a = \{u\}$, $F_d = \emptyset$  and  $B_d = \emptyset$. \\
					02. \> {\sc While} $F_a \neq \emptyset$ AND $B_a \neq \emptyset$: \\
					03. \>  \ \ \ \ \ \ \  \= $x = \arg \min_{x' \in F_a} \{ k(x) \}$. \\
					04. \> \> {\sc If} $k(x) \succ \min_{y' \in B_d} \{ k(y') \}$, {\sc Then} \\
					05. \> \> \ \ \ \ \ \ \ \ \= {\sc Return} NO. \qquad // Insertion of the edge $(u,v)$ does not create a cycle. \\
					06. \> \> {\sc Else} \\
					07. \> \> \> EXPLORE-FORWARD($x$). \\
					09. \>  \> $y = \arg \max_{y' \in B_a} \{ k(y') \}$. \\
					10. \> \> {\sc If} $k(y) \succ \max_{x' \in B_d} \{ k(x') \}$, {\sc Then} \\
					11. \> \> \> {\sc Return} NO. \qquad // Insertion of the edge $(u,v)$ does not create a cycle. \\
					12. \> \> {\sc Else} \\
					13. \> \> \> EXPLORE-BACKWARD($y$). \\
					14. \> {\sc Return} NO. \qquad \qquad  \qquad // Insertion of the edge $(u,v)$ does not create a cycle.
					\end{tabbing}
			\end{minipage}
	}}
	\caption{\label{app:fig:search} Subroutine:  SEARCH$(u, v)$.}
\end{figure}

\begin{figure}[h!]
	\centerline{\framebox{
			\begin{minipage}{5.5in}
				\begin{tabbing}
					1. \   \=  $F_a = F_a \setminus \{x\}$ and  $F_d = F_d \cup \{x\}$. \\
					2. \> {\sc For all} $(x,x') \in E$ with $V(x) = V(x')$: \\
					3. \> \ \ \ \ \ \ \ \= {\sc If} $x' \in B_a \cup B_d$, {\sc Then} \\
					4. \> \> \ \ \ \ \ \ \ \ \  \= {\sc Retrun} YES. \qquad // Insertion of the edge $(u,v)$ creates a cycle. \\
					5. \> \> {\sc Else if} $x' \notin F_a \cup F_d$, {\sc Then} \\
					6. \> \> \> $F_a = F_a \cup \{x'\}$. 
					\end{tabbing}
			\end{minipage}
	}}
	\caption{\label{app:fig:forward} Subroutine:  EXPLORE-FORWARD$(x)$.}
\end{figure}

\begin{figure}[h!]
	\centerline{\framebox{
			\begin{minipage}{5.5in}
				\begin{tabbing}
					1. \   \=  $B_a = B_a \setminus \{y\}$ and  $B_d = B_d \cup \{y\}$. \\
					2. \> {\sc For all} $(y', y) \in E$ with $V(y') = V(y)$: \\
					3. \> \ \ \ \ \ \ \ \= {\sc If} $y' \in F_a \cup F_d$, {\sc Then} \\
					4. \> \> \ \ \ \ \ \ \ \ \  \= {\sc Retrun} YES. \qquad // Insertion of the edge $(u,v)$ creates a cycle. \\
					5. \> \> {\sc Else if} $y' \notin B_a \cup B_d$, {\sc Then} \\
					6. \> \> \> $B_a = B_a \cup \{y'\}$. 
					\end{tabbing}
			\end{minipage}
	}}
	\caption{\label{app:fig:backward} Subroutine:  EXPLORE-BACKWARD$(y)$.}
\end{figure}

We need to specify two more aspects of the subroutine SEARCH($u,v$). First, how frequently do we explore forward nodes compared to exploring backward nodes? Second, suppose that we are going to explore a forward (resp. backward) node at the present moment. Then how do we select the node $x$ from the set $F_a$ (resp. $B_a$) that has to be explored? Below, we  state two crucial properties of our algorithm that address these two questions. See Figure~\ref{app:fig:search} for the pseudocode. 
\begin{property}(Balanced Search)
\label{app:prop:balance}
We  alternate between calls to  EXPLORE-FORWARD(.) and  EXPLORE-BACKWARD(.).  This ensures that $|B_d| - 1 \leq |F_d| \leq  |B_d|+1$  at every point in time. In other words,  every forward-exploration step is followed by a backward-exploration step and vice versa.
\end{property}

\begin{property}(Ordered Search)
\label{app:prop:order}
While deciding which node in $F_a$ to explore next, we always pick the node $x \in F_a$ that has {\em minimum} priority $k(x)$. Thus,  we ensure that the subroutine EXPLORE-FORWARD($x$) is only called on the node $x$ that appears before every other node in $F_a$ in the total ordering $\prec$. In contrast, while deciding which node in $B_a$ to explore next, we always pick the node $y \in B_a$ that has {\em maximum} priority $k(y)$. Thus,  we ensure that the subroutine EXPLORE-BACKWARD($y$) is only called on the node $x$ that appears {\em after} every other node in $B_a$ in the total ordering $\prec$.
\end{property}

An immediate consequence of Property~\ref{app:prop:order} is that there is no {\em gap} in the set $F_d$ as far as reachability from the node $v$ is concerned. To be more specific, consider the sequence of nodes in $G_{i,j}$ that are reachable from $v$ in increasing order of their positions in the total order $\prec$. This sequence starts with $v$. The set of nodes belonging to $F_d$ always form a prefix of this sequence. This observation is formally stated below.

\begin{corollary}
\label{app:cor:order:forward}
Consider any two nodes $x, y \in V_{i,j}$ such that $k(x) \prec k(y)$ and there is a path in $G_{i,j}$ from $v$ to each of these two nodes. At any point in time, if $y \in F_d$, then we must also have $x \in F_d$. 
\end{corollary}

Corollary~\ref{app:cor:order:backward} is a mirror image of Corollary~\ref{app:cor:order:forward}, albeit from the perspective of the node $u$.

\begin{corollary}
\label{app:cor:order:backward}
Consider any two nodes $x, y \in V_{i,j}$ such that $k(x) \prec k(y)$ and there is a path in $G_{i,j}$ from each of these two nodes to $u$. At any point in time, if $x \in B_d$, then we must also have $y \in B_d$. 
\end{corollary}

To complete the description of the subroutine SEARCH($u,v$), we now specify six {\em terminating conditions}. Whenever one of these conditions is satisfied, the subroutine does not need to run any further because it already  knows whether or not the insertion of the edge $(u, v)$ creates a cycle in the graph $G$.

\medskip
\noindent
{\bf (C1)} {\em  $F_a = \emptyset$}. 

\smallskip
\noindent
Here, as described in step 2 and step 14 in Figure~\ref{app:fig:search}, we conclude that the graph $G$ remains acyclic even after the insertion of the edge $(u,v)$. We now justify this conclusion.  Recall that if the insertion of the edge $(u,v)$ creates a cycle, then that cycle must contain a path $P_{v, u}$ from $v$ to $u$ in $G_{i,j}$. When the subroutine SEARCH($u,v$) begins execution, we have $F_a = \{v\}$ and $B_a = \{u\}$. Hence, Property~\ref{app:prop:node:move} implies that  at the present moment $v \in F_d \cup F_a$ and $u \in B_d \cup B_a$. Since the sets $F_d, F_a, B_d, B_a$ are pairwise mutually exclusive (see Property~\ref{app:prop:forward:back}) and  $F_a = \emptyset$, we currently have $v \in F_d$ and $u \notin F_d$. Armed with this observation, we consider the path $P_{vu}$ from $v$ to $u$, and let $x$ be the first node in this path that does not belong to  $F_d$. Let $y$ denote the node that appears just before $x$ in this path. Then by definition, we have $y \in F_d$ and $(y, x) \in E_{i,j}$. Now, applying Corollary~\ref{app:cor:alive:dead}, we get $x \in F_d \cup F_a = F_d$, which leads to a contradiction.

\medskip
\noindent
{\bf (C2)} {\em  $B_a = \emptyset$}.

\smallskip
\noindent
Here, as described in step 2 and step 14 in Figure~\ref{app:fig:search}, we conclude that the graph $G$ remains acyclic even after the insertion of the edge $(u,v)$. The argument here is analogous to the argument for condition (C1) above.

\medskip
\noindent 
 {\bf (C3)} {\em While exploring a node $x \in F_a$, we discover that $x$ has an outgoing edge to a node $x' \in B_a \cup B_d$.}

\smallskip
\noindent
Here,  we conclude that the insertion of the edge $(u,v)$ creates a cycle (see step 4 of Figure~\ref{app:fig:forward}). We now justify this conclusion. Since $x \in F_a$, Property~\ref{app:prop:forward:back}  implies that there is a path $P_{v, x}$ from $v$ to $x$. Since $x' \in B_a \cup B_d$, Property~\ref{app:prop:forward:back} also implies that there is a path $P_{x', u}$ from $x'$ to $u$. We get a cycle by combining the path $P_{v, x}$, the edge $(x, x')$,  the path $P_{x', u}$ and the edge $(u,v)$.

\medskip
\noindent
{\bf (C4)} {\em While exploring a node $y \in B_a$, we discover that $y$ has an incoming edge from a node $y' \in F_a \cup F_d$.}

\smallskip
\noindent
Here, as described in step 4 of Figure~\ref{app:fig:backward},  we conclude that the insertion of the edge $(u,v)$ creates a cycle.  To justify this conclusion, we can apply an argument  analogous to the one used in condition (C3) above. 

\medskip
\noindent
 {\bf (C5)} {\em  $\min_{x \in F_a} k(x) \succ \min_{y \in B_d} k(y)$.}

\smallskip
\noindent
Here, as described in step 4 of Figure~\ref{app:fig:search}, we conclude that the graph $G$ remains acyclic even after the insertion of the edge $(u, v)$. 
We now justify this conclusion. Suppose that the insertion of the edge $(u,v)$ creates a cycle. Such a cycle defines a path $P_{v, u}$ from $v$ to $u$. Below, we make a  claim that will be proved later on.

\begin{claim}
\label{app:cl:path}
The path $P_{v, u}$ contains at least one node $x$ from the set $F_a$.
\end{claim}

\noindent Armed with Claim~\ref{app:cl:path}, we consider any node $x'$ in the path $P_{v,u}$ that belongs to the set $F_a$. Let $y' = \arg \min_{y \in B_d} \{ k(y) \}$.  Note that $k(y') = \min_{y \in B_a} k(y) \prec \min_{x \in F_a} k(x) \preceq  k(x')$. In particular, we infer that $k(y') \prec k(x')$. As  $y' \in B_d$, the node $u$ is reachable from  $y'$ (see Property~\ref{app:prop:forward:back}). Similarly, as the node $x'$ lies on the path $P_{v, u}$, the node $u$ is also reachable from $x'$. Since the node $u$ is reachable from both  the nodes $y' \in B_d$ and $x'$,  and since $k(y') \prec k(x)$, Corollary~\ref{app:cor:order:backward} implies that $x' \in B_d$. This leads to a contradiction, for $x' \in F_a$ and $F_a \cap B_d = \emptyset$ (see Property~\ref{app:prop:forward:back}). Hence, our initial assumption was wrong, and  the insertion of the edge $(u,v)$ does not create a cycle in $G$. It now remains to prove Claim~\ref{app:cl:path}.

\smallskip
\noindent {\em Proof of Claim~\ref{app:cl:path}.}
Applying the same argument used to justify condition (C1), we first observation that $v \in F_a \cup F_d$ and $u \in B_a \cup B_d$. As the subsets $F_a, F_d, B_a$ and $B_d$ are pairwise mutually exclusive (see Property~\ref{app:prop:forward:back}), we have $u \notin F_a \cup F_d$. Note that  if $v \in F_a$, then there is nothing further to prove.  Accordingly, for the rest of the proof we consider the scenario where $v \in F_d$. Since $v \in F_d$ and $u \notin F_d$, there has to be at least one node in the path $P_{v,u}$ that does not belong to the set $F_d$.  Let $x$ be the first such node, and let $y$ be the node that appears just before $x$ in the path $P_{v, u}$. Thus, we have $y \in F_d$, $x \notin F_d$ and $(y, x) \in E_{i,j}$. Hence, Corollary~\ref{app:cor:alive:dead} implies that $x \in F_a$. So the path $P_{v, u}$ contains some node from the set $F_a$.
\qed

\medskip
\noindent
{\bf (C6)} {\em  $\max_{y \in B_a} k(y) \prec \max_{x \in F_d} k(x)$.}

\smallskip
\noindent
Here, as described in step 10 of Figure~\ref{app:fig:search}, we conclude that the graph $G$ remains acyclic. To justify this conclusion, we can apply an argument  analogous to the one used in condition (C5) above. \\

We now state an important corollary that follows from our stopping  conditions (C5) and (C6). It states that every node $x \in F_d$ appears before every node $y \in B_d$ in the total order $\prec$ in phase III.

\begin{corollary}
\label{app:cor:order:terminate}
We always have $\max_{x \in F_d} \{k(x)\} \prec \min_{y \in B_d} \{k(y)\}$.
\end{corollary}

\begin{proof}
Suppose that the corollary is false. Note that initially when the subroutine SEARCH($u, v$) begins execution, we have $F_d = B_d = \emptyset$ and hence the corollary is vacuously true at that time. Consider the first time-instant (say) $t$ when the corollary becomes false. Accordingly, we have:
\begin{equation}
\label{app:eq:time}
\max_{x \in F_d} \{k(x)\} \prec \min_{y \in B_d} \{k(y)\} \text{ just before time } t.
\end{equation}
One the following two events must have occurred at time $t$ for the corollary to get violated. 
\begin{itemize}
\item (1) A  node $x' \in F_a$   was explored during a call to  EXPLORE-FORWARD($x'$), and it happened to be the case that $k(x') \succ \min_{y \in B_d} \{k(y)\}$. The subroutine EXPLORE-FORWARD($x'$) then moved the node $x'$ from the set $F_a$ to the set $F_d$ (see step 1 in Figure~\ref{app:fig:forward}), which violated the corollary. Note that a call to the subroutine EXPLORE-FORWARD(.) can only be made from step 7 in Figure~\ref{app:fig:search}. Hence, according to steps 4-5  in Figure~\ref{app:fig:search} we had $k(x') \prec \min_{y \in B_d} \{k(y)\}$ just before time $t$. Thus, from~(\ref{app:eq:time}) we conclude that the corollary  remains satisfied even after adding the node $x'$ to the set $F_d$. This leads to a contradiction.
\item (2) A  node $x' \in F_a$   was explored during a call to  EXPLORE-FORWARD($x'$), and it happened to be the case that $k(x') \succ \min_{y \in B_d} \{k(y)\}$. The subroutine EXPLORE-FORWARD($x'$) then moved the node $x'$ from the set $F_a$ to the set $F_d$ (see step 1 in Figure~\ref{app:fig:forward}), which violated the corollary. Applying an argument analogous to the one applied in case (1), we again reach a contradiction.
\end{itemize}
Thus, our initial assumption was wrong, and we infer that the corollary is always satisfied.
\end{proof}

The next lemma states that in phase III our algorithm correctly detects a cycle. The proof of the lemma follows immediately from the preceding discussion.

\begin{lemma}
\label{app:lm:phaseIII:correct}
The subroutine SEARCH($u,v$) in Figure~\ref{app:fig:search} returns YES if the insertion of the edge $(u,v)$ creates a cycle in the graph $G$, and NO otherwise.
\end{lemma}

The lemma below bounds the time spent in any single call to the subroutine SEARCH($u, v$).

\begin{lemma}
\label{app:lm:search:time}
Consider any call to the subroutine SEARCH($u,v$). The time spent on this call is at most $\tilde{O}(m/n)$ times the size of the set $F_d$ at the end of the call.
\end{lemma}

\begin{proof}(Sketch)
Each call to  EXPLORE-FORWARD($x$) or EXPLORE-BACKWARD($x$) takes time proportional to the out-degree (resp. in-degree) of $x$ in the subgraph $G_{i,j}$. Under Assumption~\ref{app:assume}, the maximum in-degree and maximum out-degree of a node in $G_{i,j}$ are both at most $O(m/n)$. Thus, a single call to  EXPLORE-FORWARD($x$) or EXPLORE-BACKWARD($x$) takes $O(m/n)$ time. 

Now, consider the {\sc While} loop in Figure~\ref{app:fig:search}. Steps 3 and 9 in Figure~\ref{app:fig:search} can be implemented in $\tilde{O}(1)$ time by respectively using a min-heap and a max-heap data structure. So the time spent on this {\sc While} loop is dominated by the time taken by the calls to EXPLORE-FORWARD(.) and EXPLORE-BACKWARD(.). Each such call takes $O(m/n)$ time. Furthermore, after each call to EXPLORE-FORWAD(.) or EXPLORE-BACKWARD(.), the size of the set $F_d$ or $B_d$ respectively increases by one (see steps 1 in Figure~\ref{app:fig:forward} and Figure~\ref{app:fig:backward}). Accordingly, the time spent on one call to the subroutine SEARCH($u,v$) is at most $\tilde{O}(m/n)$ times the size of the set $F_d \cup B_d$ at the end of the call. The lemma now follows from Property~\ref{app:prop:balance}.
\end{proof}

\medskip
\noindent {\bf Bounding the total time spent in phase III.} We now analyze the total time spent by our algorithm in phase III, over the entire sequence of edge insertions in $G$.
For $l \in [1, m]$, consider the $l^{th}$ edge-insertion in the graph $G$, and let $t_l$ denote the size of the set $F_d$ at the end of phase III while handling this $l^{th}$-edge insertion. Lemma~\ref{app:lm:search:time} implies that the total time spent in phase III  is at most $\tilde{O}\left((m/n) \cdot \sum_{l=1}^m t_l\right)$. We now focus on upper bounding the sum $\sum_{l=1}^m t_l$. We start with the lemma below.

\begin{lemma}
\label{app:lm:search:time:total}
We have $\sum_{l=1}^m t_l^2 = O(n \tau)$. 
\end{lemma}

\begin{proof}
For any $l \in [1, m]$, and let $F_d^{(l)}$ and $B_d^{(l)}$ respectively denote the sets $F_d$ and $B_d$ at the end of phase III while handling the $l^{th}$ edge-insertion in $G$. Furthermore, let $G^{(l)}$ and $G_{i,j}^{(l)}$ respectively denote the input graph $G$ and the subgraph $G_{i,j}$  after the $l^{th}$ edge-insertion in $G$. 

Suppose that the edge $(u,v)$ is the $l^{th}$ edge to be inserted into $G$. We focus on the procedure for handling this edge insertion. During this procedure,  if we find  $k(u) \prec k(v)$ in the beginning of phase III, then our algorithm immediately declares that the insertion of the edge $(u,v)$ does not create a cycle and moves on to phase IV. In such a scenario, we clearly have $F_d^{(l)} = B_d^{(l)} = \emptyset$ and hence $t_l = 0$. Accordingly, from now on we assume that $k(v) \prec k(u)$ in the beginning of phase III. Consider any two nodes $x \in F_d^{(l)}$ and $y \in B_d^{(l)}$.  The nodes $x$ and $y$ belong to the same subgraph $G_{i,j}^{(l)}$. Property~\ref{app:prop:forward:back} guarantees that there is a path $P_{y, x}$ from $y$ to $x$ in $G_{i,j}^{(l)}$ -- we can go from $y$ to $u$, take the edge $(u,v)$ and then go from $v$ to $x$. Hence, from Lemma~\ref{app:lm:partition:related:nodes} we infer that the ordered pair $(y, x)$ is $\tau$-related in $G^{(l)}$ with high probability, and we condition on this event for the rest of the proof.  We now claim that there was no path from $y$ to $x$ in the graph $G^{(l-1)}$, where $G^{(l-1)}$ denotes the graph $G$ just before the $l^{th}$ edge-insertion, or equivalently, just after the $(l-1)^{th}$ edge-insertion. To see why this claim is true, we recall Lemma~\ref{app:lm:phase:order:correct}. This lemma states that in the beginning of phase III (after the $l^{th}$ edge-insertion) the total order $\prec$ on the node-set $V$ is a topological order in the graph $G^{(l-1)}$. Since $y \in B_d^{(l)}$ and $x \in F_d^{(l)}$, Corollary~\ref{app:cor:order:terminate} implies that $x$ appears before $y$ in the total order $\prec$ in phase III (after the $l^{th}$ edge-insertion). From these last two observations, we conclude that there is no path from $y$ to $x$ in $G^{(l-1)}$. As edges only get inserted into $G$ with the passage of time, this also implies that there is no path from $y$ to $x$ in the graph $G^{(l')}$, for all $l' < l$. Accordingly, the ordered pair $(y,x)$ is {\em not} $\tau$-related in the graph $G^{(l')}$ for any $l' < l$.

To summarize, for every node $x \in F_d^{(l)}$ and every node $y \in B_d^{(l)}$ the following conditions hold. (1) The ordered pair $(y,x)$ is $\tau$-related in the graph $G^{(l)}$. (2) For all $l' < l$, the ordered pair $(y,x)$ is {\em not} $\tau$-related in the graph $G^{(l')}$. Let $C$ denote a counter which keeps track of the number of sometime $\tau$-related pairs of nodes (see Definition~\ref{app:def:related:nodes:time}). Conditions (1) and (2) imply that every ordered pair of nodes $(y, x)$, where $y \in B_d^{(l)}$ and $x \in F_d^{(l)}$, contributes one towards the counter $C$.  A simple counting argument gives us: 
\begin{equation}
\label{app:eq:count:new}
\sum_{l=1}^m \left| F_d^{(l)} \right| \cdot \left| B_d^{(l)} \right| \leq C = O(n \tau)
\end{equation}
In the above derivation,  the last equality follows from Theorem~\ref{app:th:bound:related:nodes}. We now recall Property~\ref{app:prop:balance}, which says that our algorithm in phase III explores (almost) the same number of forward and backward nodes. In particular, we have $\left| F_d^{(l)} \right| \cdot \left| B_d^{(l)} \right| = O\left(\left| F_d^{(l)} \right|^2\right) = O(t_l^2)$ for all $l \in [1, m]$. This observation, along with~(\ref{app:eq:count:new}), implies that $\sum_{l=1}^m t_l^2 = O(n \tau)$. This concludes the proof of the lemma.
\end{proof}

\begin{corollary}
\label{app:cor:search:time:total}
We have $\sum_{l=1}^m t_l = O(\sqrt{mn\tau})$.
\end{corollary}

\begin{proof}
We partition the set of indices $\{1, \ldots, m\}$ into two subsets: 
$$X = \left\{ l \in [1, m] : t_l \leq \sqrt{n \tau/m} \right\} \text{ and } Y = \left\{ l \in [1, m] : t_l > \sqrt{n \tau/m}\right\}.$$ 
It is easy to check that $\sum_{l \in X} t_l \leq |X| \cdot \sqrt{n \tau/m} \leq m \cdot \sqrt{n \tau/m} = \sqrt{mn\tau}$.  Accordingly, for the rest of the proof we focus on bounding the sum $\sum_{l \in Y} t_l$. Towards this end, for each $l \in Y$, we first express the quantity $t_l$ as $t_l = \sqrt{n \tau/m} + \delta_l$, where $\delta_l > 0$. Now, Lemma~\ref{app:lm:search:time:total} implies that:
\begin{equation}
\label{app:eq:last}
\sum_{l \in Y} t_l^2 = \sum_{l \in Y} \left(\sqrt{n \tau/m} + \delta_l\right)^2  = O(n \tau)
\end{equation}
We also note that:
\begin{equation}
\label{app:eq:last:1}
 \sum_{l \in Y} \left(\sqrt{n \tau/m} + \delta_l\right)^2  \geq \sum_{l \in Y} \left(\delta_l \cdot \sqrt{n \tau/m}\right) =  \sqrt{n \tau/m} \cdot \sum_{l \in Y} \delta_l
\end{equation}
From~(\ref{app:eq:last}) and~(\ref{app:eq:last:1}), we get $\sqrt{n \tau/m} \cdot \sum_{l \in Y} \delta_l = O(n \tau)$, which in turn gives us: $\sum_{l \in Y} \delta_l = O\left(\sqrt{mn\tau}\right)$. This leads to the following upper bound on the sum $\sum_{l \in Y} t_l$.
\begin{equation}
\nonumber
\sum_{l \in Y} t_l = \sum_{l \in Y} \left(\sqrt{n \tau/m} + \delta_l\right) = \sum_{l \in Y} \sqrt{n \tau/m} + \sum_{l \in Y} \delta_l \leq m \cdot  \sqrt{n \tau/m} + O\left(\sqrt{mn\tau}\right)
= O\left(\sqrt{mn\tau}\right).
\end{equation}
This concludes the proof of the corollary.
\end{proof}

We are now ready to upper bound the total time spent by our algorithm in phase III.

\begin{lemma}
\label{app:lm:search:final:total:time}
Our algorithm spends $\tilde{O}\left(\sqrt{m^3 \tau/n}\right)$ total time in phase III, over the entire sequence of edge-insertions in the graph $G$.
\end{lemma}

\begin{proof}
Lemma~\ref{app:lm:search:time} implies that the total time spent in phase III is at most $O\left( (m/n) \cdot \sum_{l=1}^m t_l \right)$. The lemma now follows from Corollary~\ref{app:cor:search:time:total}.
\end{proof}

\subsubsection{Phase IV: Ensuring that $\prec$ is a  topological ordering for $G^+$ (only when $G^+$ is acyclic)}
\label{app:sec:phase:final}

As in Section~\ref{app:sec:phase:cycle}, we let $G^-$ and $G^+$ respectively denote the graph $G$ just before and after the insertion of the edge $(u,v)$. If in phase III we detect a cycle, then we do not need to perform any nontrivial computation from this point onward, for the graph $G$ will contain a cycle after every future edge-insertion. Hence, throughout this section we assume that no cycle was detected in phase III, and as per Lemma~\ref{app:lm:phaseIII:correct} the graph $G^+$ is acyclic.  Our goal in phase IV is to update  $\prec$ so that it becomes a topological ordering in $G^+$. 

Consider two possible cases, depending on whether $u$ appeared before or after $v$ in the total order $\prec$ in the beginning of phase III. If $k(u) \prec k(v)$ in the beginning of phase III, then the discussion in the first three paragraphs of Section~\ref{app:sec:phase:cycle} implies that $\prec$ is already a topological ordering of $G^+$, and nothing further needs to be done. Thus, from now on we assume that $k(v) \prec k(u)$ and $V(u) = V(v) = V_{i,j}$ for some $i, j \in [0, |S|]$ in the beginning of phase.

Recall the six terminating conditions for the subroutine SEARCH($u, v$) used in phase III (see the discussion after Corollary~\ref{app:cor:order:backward}). We have already assumed that we do not detect any cycle in phase III. Hence, the subroutine SEARCH($u, v$) terminates under one of the following four conditions: (C1), (C2), (C5) and (C6). How we update the total order $\prec$ in phase IV depends on the terminating condition under which the subroutine SEARCH($u,v$) returned in phase III. In particular, there are two cases to consider.

\medskip
\noindent {\bf Case 1. The subroutine SEARCH($u,v$) returned under condition (C2) or (C6) in phase III.} 

\smallskip
\noindent In this scenario, we update the total order $\prec$ by calling the subroutine described in Figure~\ref{app:fig:update:forward}. In this subroutine, the symbols $F_d$ and $B_d$ respectively denotes the set of forward-dead and backward-dead nodes at the end of phase III. Similarly, we will use the symbols $F_a$ and $B_a$ respectively to denote the set of forward-alive and backward-alive nodes at the end of phase III. The subroutine works as follows.

Note that when the subroutine SEARCH($u,v$) begins execution in phase III, we had $v \in F_a$ and $u \in B_a$ (see step 1 in Figure~\ref{app:fig:search}). Since the subroutine SEARCH($u,v$) returned under conditions (C2) or (C6), Property~\ref{app:prop:node:move} implies that $v \in F_d$ and $u \in B_d$ at the end of phase III.  Accordingly, at this point in time, let $v, x_1, \ldots, x_f$ denote the nodes in $F_d$ in increasing order of their priorities, that is, we have $k(v) \prec k(x_1) \prec \cdots \prec k(x_f)$. Similarly, let $y_1, \cdots, y_b, u$ denote the nodes in $B_d$ in increasing order of their priorities, that is, we have $k(y_1) \prec \cdots \prec k(y_b) \prec k(u)$.  By Corollary~\ref{app:cor:order:terminate}, we also have $k(x_f) \prec k(y_1)$. Now that the edge $(u,v)$ has been inserted, we need to update the relative ordering among the nodes in $F_d \cup B_d$.

Steps 1-8 in Figure~\ref{app:fig:update:forward} update the total order $\prec$ in such a way  that it satisfies the following properties. (1) We still have $k(v) \prec k(x_1) \prec \cdots \prec k(x_f)$. In words, the relative ordering among the nodes in $F_d$ does not change. (2) Consider any two nodes $x, y \in V$
 such that we had $k(x) \prec k(x_f) \prec k(y)$ at the end of phase III. Then we still have $k(x) \prec k(x_f) \prec k(y)$ at the end of step 8 in Figure~\ref{app:fig:update:forward}. In words, the relative position of $x_f$ among all the nodes in $V$ does not change. (3)  The nodes in $F_d$ occur in consecutive positions in the total order $\prec$. In other words, at the end of step 8 we cannot find any node $x \notin F_d$ and any two nodes $x', x'' \in F_d$ such that $k(x') \prec k(x) \prec k(x'')$. 
 
 \begin{claim}
 \label{app:cl:critical}
 Consider any edge $(x, y)$ in $G^-$ where $y \in B_d$ and $x \notin B_d$. Then  $k(x) \prec k(v)$ at the end of step 8 in Figure~\ref{app:fig:update:forward}.
 \end{claim}
 
 \begin{proof}
Since $y \in B_d$, there is an edge from $x$ to $y$ and $x \notin B_d$, Corollary~\ref{app:cor:alive:dead} implies that $x \in B_a$. Hence, the subroutine SEARCH($u,v$) returned under condition (C6), and {\em not} under condition (C2). By condition (C6), we have $k(x) \prec k(x_f)$ at the end of phase III. Since steps 1-8 in Figure~\ref{app:fig:update:forward} ensures that the nodes in $F_d$ occur in consecutive positions in $\prec$ and does not change the relative position of $x_f$ among all the nodes in $V$, it follows that $k(x) \prec k(v)$ at the end of step 8 in Figure~\ref{app:fig:update:forward}.
 \end{proof}
 
 Steps 9-15 in Figure~\ref{app:fig:update:forward} now further update the total order $\prec$ in such a way  that it satisfies the following properties. (4) We still have $k(y_1) \prec \cdots \prec k(y_b) \prec k(u)$. In words, the relative ordering among the nodes in $B_d$ does not change. (5) The node $u$ is placed immediately before the node $v$ in the total order $\prec$. This is consistent with the fact that the edge $(u,v)$ has been inserted into the graph $G$. (6) The nodes in $B_d$ occur in consecutive positions in the total order $\prec$. In other words, at the end of step 15 we cannot find any node $y \notin B_d$ and any two nodes $y', y'' \in F_d$ such that $k(y') \prec k(y) \prec k(y'')$. 
 
 To summarize, at this point in time in  $\prec$ the nodes $y_1, \ldots, y_b, u, v, x_1, \ldots, x_f$ occur consecutive to one another, and in this order. Accordingly, Corollary~\ref{app:cor:order:forward}, Corollary~\ref{app:cor:order:backward} and Claim~\ref{app:cl:critical} ensure that the total order $\prec$ remains a topological order in $G^-$ at this point in time. Since $u$ appears before $v$ in $\prec$, we also conclude that at this point in time $\prec$ is a topological order in $G^+$.
 
 \medskip
\noindent {\bf Case 2. The subroutine SEARCH($u,v$) returned under condition (C1) or (C5) in phase III.} 

 \smallskip
 \noindent This case is completely analogous to case 1 above, and we omit its description.

\begin{lemma}
\label{app:lm:time:update:order}
We spend $\tilde{O}\left(  \sqrt{mn \tau} \right)$ time in phase IV, over the entire sequence of edge-insertions in $G$.
\end{lemma}

\begin{proof}(Sketch)
Steps 7 and 14 in Figure~\ref{app:fig:update:forward} can be implemented in $O(1)$ time using the ordered list data structure~\cite{DietzS87}. Hence, the time spent in phase IV after a given edge-insertion is proportional to the sizes of the sets $F_d$ and $B_d$ at the end of phase III, and  by Property~\ref{app:prop:balance}, the sizes of the sets $F_d$ and $B_d$ are (almost) equal to one another. For $l \in [1, m]$, let $t_l$ denote the size of the set $F_d$ at the end of phase III while handling the $l^{th}$ edge-insertion in $G$. We conclude that the total time spent in phase IV, during the entire sequence of edge-insertions in $G$, is given by $O\left(\sum_{l=1}^m t_l \right)$. The lemma now follows from Corollary~\ref{app:cor:search:time:total}.
\end{proof}

\begin{figure}[h!]
	\centerline{\framebox{
			\begin{minipage}{5.5in}
				\begin{tabbing}
					01. \  \ \=   $Q = F_d$. \\
					02. \> $x^* = \arg \max_{x \in Q} \{ k(x) \}$	 \\				
					03. \> $Q = Q \setminus \{x^*\}$. \\
					04. \> {\sc While} $Q \neq \emptyset$: \\
					05. \> \ \ \ \ \ \ \ \ \= $x' = \arg \max_{x \in Q} \{k(x)\}$. \\
					06. \> \> $Q = Q \setminus \{x'\}$. \\
					07. \> \> INSERT-BEFORE($x', x^*$). \\
					08. \> \> $x^* = x'$. \\
					09. \> $y^* = v$. \\
					10. \> $Q = B_d$. \\
					11. \> {\sc While} $Q \neq \emptyset$: \\
					12. \> \> $y' = \arg \max_{y \in Q} \{ k(y) \}$. \\
					13. \> \> $Q = Q \setminus \{ y' \}$. \\
					14. \> \> INSERT-BEFORE($y', y^*$). \\
					15. \> \> $y^* = y'$. 
					\end{tabbing}
			\end{minipage}
	}}
	\caption{\label{app:fig:update:forward} Subroutine:  UPDATE-FORWARD$(.)$.}
\end{figure}

\subsection{Bounding the Total Update Time of Our Algorithm}
\label{app:sec:time}

We simply add up the total time spent by our algorithm in each of these four phases, throughout the entire sequence of edge-insertions in $G$. In particular, we invoke Lemma~\ref{app:lm:partition:time}, Lemma~\ref{app:lm:phase:order:time}, Lemma~\ref{app:lm:search:final:total:time} and Lemma~\ref{app:lm:time:update:order} and conclude that the total expected update time of our algorithm is at most:
\begin{equation}
\label{app:eq:total:time:final}
\tilde{O}\left( mn/\tau + n^2/\tau + \sqrt{m^3\tau/n} + \sqrt{mn\tau} \right) = \tilde{O}\left( mn/\tau +  \sqrt{m^3\tau/n}  \right).
\end{equation}
In the above derivation, we have made the assumption that $m = \Omega(n)$. Now, setting $\tau = n/m^{1/3}$, we get a total expected update time of $\tilde{O}(m^{4/3})$. This concludes the proof of Theorem~\ref{th:main}.

\end{document}